\documentclass{article}
\usepackage{cite}
\usepackage{booktabs}
\usepackage{amsmath}
\usepackage{amsthm}
\usepackage{algorithmic}
\usepackage{multicol}
\usepackage{bbm}
\usepackage[pdftex]{hyperref}
\usepackage{array}
\usepackage{color}
\usepackage{listings}
\usepackage{graphicx}
\usepackage{booktabs}
\usepackage{mathrsfs}
\RequirePackage{amsmath}
\RequirePackage{amssymb} 
\usepackage{paralist}
\usepackage{mathrsfs}
\usepackage[ruled,commentsnumbered,linesnumbered]{algorithm2e}
\usepackage{amsthm}

\hypersetup{
    bookmarks=true,         % show bookmarks bar?
    unicode=false,          % non-Latin characters in Acrobat’s bookmarks
    pdftoolbar=true,        % show Acrobat’s toolbar?
    pdfmenubar=true,        % show Acrobat’s menu?
    pdffitwindow=false,     % window fit to page when opened
    pdfstartview={FitH},    % fits the width of the page to the window
    pdftitle={My title},    % title
    pdfauthor={Author},     % author
    pdfsubject={Subject},   % subject of the document
    pdfcreator={Creator},   % creator of the document
    pdfproducer={Producer}, % producer of the document
    pdfkeywords={keyword1} {key2} {key3}, % list of keywords
    pdfnewwindow=true,      % links in new window
    colorlinks=true,       % false: boxed links; true: colored links
    linkcolor=blue,          % color of internal links (change box color with linkbordercolor)
    citecolor=red,        % color of links to bibliography
    filecolor=magenta,      % color of file links
    urlcolor=blue           % color of external links
}

\newcommand{\hili}[1]{#1}

\newcommand{\zhenqi}[1]{}
\newcommand{\ep}[0]{\mathbbm{E}}

\DeclareMathOperator*{\argmin}{arg\,min}
\DeclareMathOperator*{\grad}{\bigtriangledown}
\newcommand{\pj}[1]{Proj_{#1}}

\newcommand{\sayan}[1]{}

\newcommand{\num}[1]{\relax\ifmmode \mathbb #1\else $\mathbb #1$\fi}
\newcommand{\nnnum}[1]{\relax\ifmmode 
  {\mathbb #1}_{\geq 0} \else ${\mathbb #1}_{\geq 0}$
  \fi}
\newcommand{\npnum}[1]{\relax\ifmmode 
  {\mathbb #1}_{\leq 0} \else ${\mathbb #1}_{\leq 0}$
  \fi}
\newcommand{\pnum}[1]{\relax\ifmmode 
  {\mathbb #1}_{> 0} \else ${\mathbb #1}_{> 0}$
  \fi}
\newcommand{\nnum}[1]{\relax\ifmmode 
  {\mathbb #1}_{< 0} \else ${\mathbb #1}_{< 0}$
  \fi}
\newcommand{\plnum}[1]{\relax\ifmmode 
  {\mathbb #1}_{+} \else ${\mathbb #1}_{+}$
  \fi}
\newcommand{\nenum}[1]{\relax\ifmmode 
  {\mathbb #1}_{-} \else ${\mathbb #1}_{-}$
  \fi}

\newcommand{\reals}{{\num R}}                    %reals
\newcommand{\naturals}{{\num N}}                      %natural numbers
\newcommand{\pr}{{\num P}}                      %probability measure

\newcommand{\extb}[1]{\relax\ifmmode {\sf ExtBeh}_{#1} \else ${\sf ExtBeh}_{#1}$\fi} 
\newcommand{\tdists}[1]{\relax\ifmmode {\sf Tdists}_{#1} \else ${\sf Tdists}_{#1}$\fi} 

\newcommand{\exec}[1]{\relax\ifmmode {\sf Execs}_{#1} \else ${\sf Exec}_{#1}$\fi} 
\newcommand{\execf}[1]{\relax\ifmmode {\sf Execs}^*_{#1} \else ${\sf Exec}^*_{#1}$\fi} 
\newcommand{\execi}[1]{\relax\ifmmode {\sf Execs}^\omega_{#1} \else ${\sf Exec}^\omega_{#1}$\fi} 

\newcommand{\ctrace}[1]{\relax\ifmmode {\sf Ctraces}_{#1} \else ${\sf Ctraces}_{#1}$\fi} 

\newcommand{\trace}[1]{\relax\ifmmode {\sf Traces}_{#1} \else ${\sf Traces}_{#1}$\fi} 
\newcommand{\tracef}[1]{\relax\ifmmode {\sf Traces}^*_{#1} \else ${\sf Traces}^*_{#1}$\fi} 
\newcommand{\tracei}[1]{\relax\ifmmode {\sf Traces}^\omega_{#1} \else ${\sf Traces}^\omega_{#1}$\fi} 

\newcommand{\frag}[1]{\relax\ifmmode {\sf Frags}_{#1} \else ${\sf Frags}_{#1}$\fi} 
\newcommand{\fragf}[1]{\relax\ifmmode {\sf Frags}^*_{#1} \else ${\sf Frags}^*_{#1}$\fi} 
\newcommand{\fragi}[1]{\relax\ifmmode {\sf Frags}^\omega_{#1} \else ${\sf Frags}^\omega_{#1}$\fi} 

\newcommand{\reach}[1]{\relax\ifmmode {\sf Reach}_{#1} \else ${\sf Reach}_{#1}$\fi} 
\newcommand{\pair}[2]{\relax\ifmmode \langle #1, #2 \rangle \else $\langle #1, #2 \rangle$\fi} 

\newcommand{\TE}{\relax\ifmmode \mathit{Time} \else $\mathit{Time}$ \fi} 
\newcommand{\EQ}{\relax\ifmmode \mathit{Enq} \else $\mathit{Enq}$ \fi} 
\newcommand{\DQ}{\relax\ifmmode \mathit{Deq} \else $\mathit{DeqTime}$ \fi} 
\newcommand{\E}{\relax\ifmmode \mathsf{E} \else $\mathsf{E}$ \fi}

\newcommand{\loc}{\relax\ifmmode \mathit{loc} \else $\mathit{loc}$ \fi}
\newcommand{\abs}{\relax\ifmmode \mathit{abs} \else $\mathit{abs}$ \fi}

 \newtheorem{theorem}{Theorem}
 \newtheorem{lemma}[theorem]{Lemma}

 \newtheorem{assumption}{Assumption}
 \newtheorem{definition}{Definition}
 \newtheorem{prop}[theorem]{Proposition}

  \newcounter{rem}
\newcounter{example}
 \newenvironment{example}
 {\refstepcounter{example} \vspace{2ex}\par\noindent
 \textbf{Example}~\textbf{\theexample~}}{\qed}
 
\def\A{{\cal A}} % HA
\def\E{{\cal E}} % HA
\def\F{{\cal F}} % HA
\def\G{{\cal G}} % pieces of SHIOA
\def\I{{\cal I}} % environment sequence
\def\P{{\cal P}} % set of modes
\def\R{{\cal R}} % relation
\def\T{{\cal T}} % set of trajectories
\def\V{{\cal V}} % Lyapunov function
\def\U{{\cal U}} % set of trajectories
\def\W{{\cal W}}
\def\X{{\cal X}} % Lyapunov function
\newcommand{\col}[1]{\relax\ifmmode \mathscr #1\else $\mathscr #1$\fi}

\definecolor{HIOAcolor}{rgb}{0.776,0.22,0.07}

\newcommand{\SC}[2]{\relax\ifmmode {\tt Scount}(#1,#2) \else ${\tt Scount}(#1,#2)$\fi} 
\newcommand{\SCM}[2]{\relax\ifmmode {\tt Smin}(#1,#2) \else ${\tt Smin}(#1,#2)$\fi} 
\newcommand{\Aut}[1]{\relax\ifmmode {\tt Aut}(#1) \else ${\tt Aut}(#1)$\fi}

\newcommand{\deq}{\mathrel{\stackrel{\scriptscriptstyle\Delta}{=}}}

\newcommand{\remove}[1]{}
\newcommand{\salg}[1]{\relax\ifmmode {\mathcal F}_{#1}\else ${\mathcal F}_{#1}$\fi} 
\newcommand{\msp}[1]{\relax\ifmmode (#1, \salg{#1}) \else $(#1, \salg{#1})$\fi} 
\newcommand{\msprod}[2]{\relax\ifmmode ( #1 \times #2, \salg{#1} \otimes \salg{#2}) \else $(#1 \times #2, \salg{#1} \otimes \salg{#2})$\fi} 
\newcommand{\dist}[1]{\relax\ifmmode {\mathcal P}\msp{#1}
  \else ${\mathcal P}\msp{#1}$\fi} 
\newcommand{\subdist}[1]{\relax\ifmmode {\mathcal S}{\mathcal P}\msp{#1} 
  \else ${\mathcal S}{\mathcal P}\msp{#1}$\fi} 
\newcommand{\disc}[1]{\relax\ifmmode {\sf Disc}(#1)
  \else ${\sf Disc}(#1)$\fi} 

\newcommand{\Trajeq}{\relax\ifmmode {\mathcal R}_\T \else ${\mathcal R}_\T$\fi} 
\newcommand{\Acteq}{\relax\ifmmode {\mathcal R}_A \else ${\mathcal R}_A$\fi} 
\newcommand{\noop}{\relax\ifmmode \lambda \else $\lambda$\fi} 
\newcommand{\close}[1]{\relax\ifmmode \overline{#1} \else $\overline{#1}$\fi}

\newcommand{\tup}[1]
           {
             \relax\ifmmode
             \langle #1 \rangle
             \else $\langle$ #1 $\rangle$ \fi
           }

\newcommand{\lit}[1]{ \relax\ifmmode
                \mathord{\mathcode`\-="702D\sf #1\mathcode`\-="2200}
                \else {\it #1} \fi }

\newcommand{\figuresize}{\scriptsize}

\lstdefinelanguage{ioa}{
  basicstyle=\figuresize,
  keywordstyle=\bf \figuresize,
  identifierstyle=\it \figuresize,
  emphstyle=\tt \figuresize,
  mathescape=true,
  tabsize=20,
  sensitive=false,
  columns=fullflexible,
  keepspaces=false,
  flexiblecolumns=true,
  basewidth=0.05em,
  moredelim=[il][\rm]{//},
  moredelim=[is][\sf \figuresize]{!}{!},
  moredelim=[is][\bf \figuresize]{*}{*},
  keywords={automaton,and, 
         choose,const,continue, components,
         discrete, do, derived,
         eff, external,else, elseif, evolve, end,each,
         fi,for, forward, from,
         hidden,
         in,input,internal,if,invariant, initially, imports,
     let,
     or, output, operators, od, of,
     pre, prob,
     return,
     such,satisfies, stop, signature, simulation, state, stochastic,
     trajectories,trajdef, transitions, that,then, type, types, to, tasks,
     variables, vocabulary, uni,
     when,where, with,while},
  emph={set, seq, tuple, map, array, enumeration},   
   literate=
        {(}{{$($}}1
        {)}{{$)$}}1
        {\\in}{{$\in\ $}}1
        {\\preceq}{{$\preceq\ $}}1
        {\\subset}{{$\subset\ $}}1
        {\\subseteq}{{$\subseteq\ $}}1
        {\\supset}{{$\supset\ $}}1
        {\\supseteq}{{$\supseteq\ $}}1
        {\\forall}{{$\forall$}}1
        {\\le}{{$\le\ $}}1
        {\\ge}{{$\ge\ $}}1
        {\\gets}{{$\gets\ $}}1
        {\\cup}{{$\cup\ $}}1
        {\\cap}{{$\cap\ $}}1
        {\\langle}{{$\langle$}}1
        {\\rangle}{{$\rangle$}}1
        {\\exists}{{$\exists\ $}}1
        {\\bot}{{$\bot$}}1
        {\\rip}{{$\rip$}}1
        {\\emptyset}{{$\emptyset$}}1
        {\\notin}{{$\notin\ $}}1
        {\\not\\exists}{{$\not\exists\ $}}1
        {\\ne}{{$\ne\ $}}1
        {\\to}{{$\to\ $}}1
        {\\implies}{{$\implies\ $}}1
        {<}{{$<\ $}}1
        {>}{{$>\ $}}1
        {=}{{$=\ $}}1
        {~}{{$\neg\ $}}1
        {|}{{$\mid$}}1
        {'}{{$^\prime$}}1
        {\\A}{{$\forall\ $}}1
        {\\E}{{$\exists\ $}}1
        {\\/}{{$\vee\,$}}1
        {\\vee}{{$\vee\,$}}1
        {/\\}{{$\wedge\,$}}1
        {\\wedge}{{$\wedge\,$}}1
        {=>}{{$\Rightarrow\ $}}1
        {->}{{$\rightarrow\ $}}1
        {<=}{{$\Leftarrow\ $}}1
        {<-}{{$\leftarrow\ $}}1
        {~=}{{$\neq\ $}}1
        {\\U}{{$\cup\ $}}1
        {\\I}{{$\cap\ $}}1
        {|-}{{$\vdash\ $}}1
        {-|}{{$\dashv\ $}}1
        {<<}{{$\ll\ $}}2
        {>>}{{$\gg\ $}}2
        {||}{{$\|$}}1
        {[}{{$[$}}1
        {]}{{$\,]$}}1
        {[[}{{$\langle$}}1
        {]]]}{{$]\rangle$}}1
        {]]}{{$\rangle$}}1
        {<=>}{{$\Leftrightarrow\ $}}2
        {<->}{{$\leftrightarrow\ $}}2
        {(+)}{{$\oplus\ $}}1
        {(-)}{{$\ominus\ $}}1
        {_i}{{$_{i}$}}1
        {_j}{{$_{j}$}}1
        {_{i,j}}{{$_{i,j}$}}3
        {_{j,i}}{{$_{j,i}$}}3
        {_0}{{$_0$}}1
        {_1}{{$_1$}}1
        {_2}{{$_2$}}1
        {_n}{{$_n$}}1
        {_p}{{$_p$}}1
        {_k}{{$_n$}}1
        {-}{{$\ms{-}$}}1
        {@}{{}}0
        {\\delta}{{$\delta$}}1
        {\\R}{{$\R$}}1
        {\\Rplus}{{$\Rplus$}}1
        {\\N}{{$\N$}}1
        {\\times}{{$\times\ $}}1
        {\\tau}{{$\tau$}}1
        {\\alpha}{{$\alpha$}}1
        {\\beta}{{$\beta$}}1
        {\\gamma}{{$\gamma$}}1
        {\\ell}{{$\ell\ $}}1
        {--}{{$-\ $}}1
        {\\TT}{{\hspace{1.5em}}}3        
      }

\lstdefinelanguage{ioaNums}[]{ioa}
{
  numbers=left,
  numberstyle=\tiny,
  stepnumber=2,
  numbersep=4pt
}

\lstdefinelanguage{ioaNumsRight}[]{ioa}
{
  numbers=right,
  numberstyle=\tiny,
  stepnumber=2,
  numbersep=4pt
}

\lstnewenvironment{IOA}%
  {\lstset{language=IOA}}
  {}

\lstnewenvironment{IOANums}%
  {
  \if@firstcolumn
    \lstset{language=IOA, numbers=left, firstnumber=auto}
  \else
    \lstset{language=IOA, numbers=right, firstnumber=auto}
  \fi
  }
  {}

\lstnewenvironment{IOANumsRight}%
  {
    \lstset{language=IOA, numbers=right, firstnumber=auto}
  }
  {}

\newcommand{\linefigioa}[9]{

}

\lstdefinelanguage{ioaLang}{%
  basicstyle=\ttfamily\small,
  keywordstyle=\rmfamily\bfseries\small,
  identifierstyle=\small,
  keywords={assumes,automaton,axioms,backward,bounds,by,case,choose,components,const,d,det,discrete,do,eff,else,elseif,ensuring,enumeration,evolve,fi,fire,follow,for,forward,from,hidden,if,in,%
    input,initially,internal,invariant,let, local,od,of,output,pre,schedule,signature,so,%
    simulation,states,state,variables, tasks, stop,tasks,that,then,to,trajdef,trajectory,trajectories,transitions,tuple,type,
    uniform,union,urgent,uses,when,where,while,yield},
  literate=
        {\\in}{{$\in$}}1
        {\\preceq}{{$\preceq$}}1
        {\\subset}{{$\subset$}}1
        {\\subseteq}{{$\subseteq$}}1
        {\\supset}{{$\supset$}}1
        {\\supseteq}{{$\supseteq$}}1
        {\\rho}{{$\rho$}}1
        {\\infty}{{$\infty$}}1
        {<}{{$<$}}1
        {>}{{$>$}}1
        {=}{{$=$}}1
        {~}{{$\neg$}}1 
        {|}{{$\mid$}}1
        {'}{{$^\prime$}}1
        {\\A}{{$\forall$}}1 {\\E}{{$\exists$}}1
        {\\/}{{$\vee$}}1 {/\\}{{$\wedge$}}1 
        {=>}{{$\Rightarrow$}}1 
        {->}{{$\rightarrow$}}1 
        {<=}{{$\leq$}}1 {>=}{{$\geq$}}1 {~=}{{$\neq$}}1
        {\\U}{{$\cup$}}1 {\\I}{{$\cap$}}1
        {|-}{{$\vdash$}}1 {-|}{{$\dashv$}}1
        {<<}{{$\ll$}}2 {>>}{{$\gg$}}2
        {||}{{$\|$}}1
        {<=>}{{$\Leftrightarrow$}}2 
        {<->}{{$\leftrightarrow$}}2
        {(+)}{{$\oplus$}}1
        {(-)}{{$\ominus$}}1
}

\lstdefinelanguage{bigIOALang}{%
  basicstyle=\ttfamily,
  keywordstyle=\rmfamily\bfseries,
  identifierstyle=,
  keywords={assumes,automaton,axioms,backward,by,case,choose,components,const,%
    d,det,discrete,do,eff,else,elseif,ensuring,enumeration,evolve,fi,for,forward,from,hidden,if,in%
    input,initially,internal,invariant,local,od,of,output,pre,schedule,signature,so,%
    tasks, simulation,states,stop,tasks,that,then,to,trajdef,trajectories,transitions,tuple,type,union,urgent,uses,when,where,yield},
  literate=
        {\\in}{{$\in$}}1
        {\\preceq}{{$\preceq$}}1
        {\\subset}{{$\subset$}}1
        {\\subseteq}{{$\subseteq$}}1
        {\\supset}{{$\supset$}}1
        {\\supseteq}{{$\supseteq$}}1
        {<}{{$<$}}1
        {>}{{$>$}}1
        {=}{{$=$}}1
        {~}{{$\neg$}}1 
        {|}{{$\mid$}}1
        {'}{{$^\prime$}}1
        {\\A}{{$\forall$}}1 {\\E}{{$\exists$}}1
        {\\/}{{$\vee$}}1 {/\\}{{$\wedge$}}1 
        {=>}{{$\Rightarrow$}}1 
        {->}{{$\rightarrow$}}1 
        {<=}{{$\leq$}}1 {>=}{{$\geq$}}1 {~=}{{$\neq$}}1
        {\\U}{{$\cup$}}1 {\\I}{{$\cap$}}1
        {|-}{{$\vdash$}}1 {-|}{{$\dashv$}}1
        {<<}{{$\ll$}}2 {>>}{{$\gg$}}2
        {||}{{$\|$}}1
        {<=>}{{$\Leftrightarrow$}}2 
        {<->}{{$\leftrightarrow$}}2
        {(+)}{{$\oplus$}}1
        {(-)}{{$\ominus$}}1
}

\lstnewenvironment{BigIOA}%
  {\lstset{language=bigIOALang,basicstyle=\ttfamily}
   \csname lst@SetFirstLabel\endcsname}
  {\csname lst@SaveFirstLabel\endcsname\vspace{-4pt}\noindent}

\lstnewenvironment{SmallIOA}%
  {\lstset{language=ioaLang,basicstyle=\ttfamily\scriptsize}
   \csname lst@SetFirstLabel\endcsname}
  {\csname lst@SaveFirstLabel\endcsname\noindent}

\newcommand{\true}{\relax\ifmmode \mathit true \else \em true \/\fi}
\newcommand{\false}{\relax\ifmmode \mathit false \else \em false \/\fi}

\newlength{\bracklen}

\newcommand{\tri}[3]{\ensuremath{\mathit{#1}^\mathit{#2}_\mathit{#3}}}

\newcommand{\sugLocalVars}[2]{\ifthenelse{\equal{}{#2}}%
                             {\tri{localVars}{#1}{desug}}%
                             {\tri{localVars}{#1}{#2,desug}}}
\newcommand{\sugVars}[2]{\ifthenelse{\equal{}{#2}}%
                        {\tri{vars}{#1}{desug}}%
                        {\tri{vars}{#1}{#2,desug}}}

\newenvironment{subSyntax}{\begin{array}{l}}{\end{array}}

\newcommand{\ms}[1]{\ifmmode%
\mathord{\mathcode`-="702D\it #1\mathcode`\-="2200}\else%
$\mathord{\mathcode`-="702D\it #1\mathcode`\-="2200}$\fi}

\def\A{{\cal A}} % TA
\def\T{{\cal T}} % set of trajectories

\lstdefinelanguage{pvs}{
  basicstyle=\tt \figuresize,
  keywordstyle=\sc \figuresize,
  identifierstyle=\it \figuresize,
  emphstyle=\tt \figuresize,
  mathescape=true,
  tabsize=20,
  sensitive=false,
  columns=fullflexible,
  keepspaces=false,
  flexiblecolumns=true,
  basewidth=0.05em,
  moredelim=[il][\rm]{//},
  moredelim=[is][\sf \figuresize]{!}{!},
  moredelim=[is][\bf \figuresize]{*}{*},
  keywords={and, 
         begin,
         cases, const,
         do,
         external, else, exists, end, endcases, endif,
         fi,for, forall, from,
         hidden,
         in, if, importing,
     let, lambda, lemma,
     measure, 
     not,
     or, of,
     return, recursive,
     stop, 
     theory, that,then, type, types, type+, to, theorem,
     var,
     with,while},
  emph={nat, setof, sequence, eq, tuple, map, array, enumeration, bool, real, exp, nnreal, posreal},   
   literate=
        {(}{{$($}}1
        {)}{{$)$}}1
        {\\in}{{$\in\ $}}1
        {\\mapsto}{{$\rightarrow\ $}}1
        {\\preceq}{{$\preceq\ $}}1
        {\\subset}{{$\subset\ $}}1
        {\\subseteq}{{$\subseteq\ $}}1
        {\\supset}{{$\supset\ $}}1
        {\\supseteq}{{$\supseteq\ $}}1
        {\\forall}{{$\forall$}}1
        {\\le}{{$\le\ $}}1
        {\\ge}{{$\ge\ $}}1
        {\\gets}{{$\gets\ $}}1
        {\\cup}{{$\cup\ $}}1
        {\\cap}{{$\cap\ $}}1
        {\\langle}{{$\langle$}}1
        {\\rangle}{{$\rangle$}}1
        {\\exists}{{$\exists\ $}}1
        {\\bot}{{$\bot$}}1
        {\\rip}{{$\rip$}}1
        {\\emptyset}{{$\emptyset$}}1
        {\\notin}{{$\notin\ $}}1
        {\\not\\exists}{{$\not\exists\ $}}1
        {\\ne}{{$\ne\ $}}1
        {\\to}{{$\to\ $}}1
        {\\implies}{{$\implies\ $}}1
        {<}{{$<\ $}}1
        {>}{{$>\ $}}1
        {=}{{$=\ $}}1
        {~}{{$\neg\ $}}1
        {|}{{$\mid$}}1
        {'}{{$^\prime$}}1
        {\\A}{{$\forall\ $}}1
        {\\E}{{$\exists\ $}}1
        {\\/}{{$\vee\,$}}1
        {\\vee}{{$\vee\,$}}1
        {/\\}{{$\wedge\,$}}1
        {\\wedge}{{$\wedge\,$}}1
        {->}{{$\rightarrow\ $}}1
        {=>}{{$\Rightarrow\ $}}1
        {->}{{$\rightarrow\ $}}1
        {<=}{{$\Leftarrow\ $}}1
        {<-}{{$\leftarrow\ $}}1
        {~=}{{$\neq\ $}}1
        {\\U}{{$\cup\ $}}1
        {\\I}{{$\cap\ $}}1
        {|-}{{$\vdash\ $}}1
        {-|}{{$\dashv\ $}}1
        {<<}{{$\ll\ $}}2
        {>>}{{$\gg\ $}}2
        {||}{{$\|$}}1
        {[}{{$[$}}1
        {]}{{$\,]$}}1
        {[[}{{$\langle$}}1
        {]]]}{{$]\rangle$}}1
        {]]}{{$\rangle$}}1
        {<=>}{{$\Leftrightarrow\ $}}2
        {<->}{{$\leftrightarrow\ $}}2
        {(+)}{{$\oplus\ $}}1
        {(-)}{{$\ominus\ $}}1
        {_i}{{$_{i}$}}1
        {_j}{{$_{j}$}}1
        {_{i,j}}{{$_{i,j}$}}3
        {_{j,i}}{{$_{j,i}$}}3
        {_0}{{$_0$}}1
        {_1}{{$_1$}}1
        {_2}{{$_2$}}1
        {_n}{{$_n$}}1
        {_p}{{$_p$}}1
        {_k}{{$_n$}}1
        {-}{{$\ms{-}$}}1
        {@}{{}}0
        {\\delta}{{$\delta$}}1
        {\\R}{{$\R$}}1
        {\\Rplus}{{$\Rplus$}}1
        {\\N}{{$\N$}}1
        {\\times}{{$\times\ $}}1
        {\\tau}{{$\tau$}}1
        {\\alpha}{{$\alpha$}}1
        {\\beta}{{$\beta$}}1
        {\\gamma}{{$\gamma$}}1
        {\\ell}{{$\ell\ $}}1
        {--}{{$-\ $}}1
        {\\TT}{{\hspace{1.5em}}}3        
      }

\lstdefinelanguage{BigPVS}{
  basicstyle=\tt,
  keywordstyle=\sc,
  identifierstyle=\it,
  emphstyle=\tt ,
  mathescape=true,
  tabsize=20,
  sensitive=false,
  columns=fullflexible,
  keepspaces=false,
  flexiblecolumns=true,
  basewidth=0.05em,
  moredelim=[il][\rm]{//},
  moredelim=[is][\sf \figuresize]{!}{!},
  moredelim=[is][\bf \figuresize]{*}{*},
  keywords={and, 
         begin,
         cases, const,
         do, datatype,
         external, else, exists, end, endif, endcases,
         fi,for, forall, from,
         hidden,
         in, if, importing,
     let, lambda, lemma,
     measure,
     not,
     or, of,
     return, recursive,
     stop, 
     theory, that,then, type, types, type+, to, theorem,
     var,
     with,while},
  emph={nat, setof, sequence, eq, tuple, map, array, first, rest, add, enumeration, bool, real, posreal, nnreal},   
   literate=
        {(}{{$($}}1
        {)}{{$)$}}1
        {\\in}{{$\in\ $}}1
        {\\mapsto}{{$\rightarrow\ $}}1
        {\\preceq}{{$\preceq\ $}}1
        {\\subset}{{$\subset\ $}}1
        {\\subseteq}{{$\subseteq\ $}}1
        {\\supset}{{$\supset\ $}}1
        {\\supseteq}{{$\supseteq\ $}}1
        {\\forall}{{$\forall$}}1
        {\\le}{{$\le\ $}}1
        {\\ge}{{$\ge\ $}}1
        {\\gets}{{$\gets\ $}}1
        {\\cup}{{$\cup\ $}}1
        {\\cap}{{$\cap\ $}}1
        {\\langle}{{$\langle$}}1
        {\\rangle}{{$\rangle$}}1
        {\\exists}{{$\exists\ $}}1
        {\\bot}{{$\bot$}}1
        {\\rip}{{$\rip$}}1
        {\\emptyset}{{$\emptyset$}}1
        {\\notin}{{$\notin\ $}}1
        {\\not\\exists}{{$\not\exists\ $}}1
        {\\ne}{{$\ne\ $}}1
        {\\to}{{$\to\ $}}1
        {\\implies}{{$\implies\ $}}1
        {<}{{$<\ $}}1
        {>}{{$>\ $}}1
        {=}{{$=\ $}}1
        {~}{{$\neg\ $}}1
        {|}{{$\mid$}}1
        {'}{{$^\prime$}}1
        {\\A}{{$\forall\ $}}1
        {\\E}{{$\exists\ $}}1
        {\\/}{{$\vee\,$}}1
        {\\vee}{{$\vee\,$}}1
        {/\\}{{$\wedge\,$}}1
        {\\wedge}{{$\wedge\,$}}1
        {->}{{$\rightarrow\ $}}1
        {=>}{{$\Rightarrow\ $}}1
        {->}{{$\rightarrow\ $}}1
        {<=}{{$\Leftarrow\ $}}1
        {<-}{{$\leftarrow\ $}}1
        {~=}{{$\neq\ $}}1
        {\\U}{{$\cup\ $}}1
        {\\I}{{$\cap\ $}}1
        {|-}{{$\vdash\ $}}1
        {-|}{{$\dashv\ $}}1
        {<<}{{$\ll\ $}}2
        {>>}{{$\gg\ $}}2
        {||}{{$\|$}}1
        {[}{{$[$}}1
        {]}{{$\,]$}}1
        {[[}{{$\langle$}}1
        {]]]}{{$]\rangle$}}1
        {]]}{{$\rangle$}}1
        {<=>}{{$\Leftrightarrow\ $}}2
        {<->}{{$\leftrightarrow\ $}}2
        {(+)}{{$\oplus\ $}}1
        {(-)}{{$\ominus\ $}}1
        {_i}{{$_{i}$}}1
        {_j}{{$_{j}$}}1
        {_{i,j}}{{$_{i,j}$}}3
        {_{j,i}}{{$_{j,i}$}}3
        {_0}{{$_0$}}1
        {_1}{{$_1$}}1
        {_2}{{$_2$}}1
        {_n}{{$_n$}}1
        {_p}{{$_p$}}1
        {_k}{{$_n$}}1
        {-}{{$\ms{-}$}}1
        {@}{{}}0
        {\\delta}{{$\delta$}}1
        {\\R}{{$\R$}}1
        {\\Rplus}{{$\Rplus$}}1
        {\\N}{{$\N$}}1
        {\\times}{{$\times\ $}}1
        {\\tau}{{$\tau$}}1
        {\\alpha}{{$\alpha$}}1
        {\\beta}{{$\beta$}}1
        {\\gamma}{{$\gamma$}}1
        {\\ell}{{$\ell\ $}}1
        {--}{{$-\ $}}1
        {\\TT}{{\hspace{1.5em}}}3        
      }

\lstdefinelanguage{pvsNums}[]{pvs}
{
  numbers=left,
  numberstyle=\tiny,
  stepnumber=2,
  numbersep=4pt
}

\lstdefinelanguage{pvsNumsRight}[]{pvs}
{
  numbers=right,
  numberstyle=\tiny,
  stepnumber=2,
  numbersep=4pt
}

\lstnewenvironment{BigPVS}%
  {\lstset{language=BigPVS}}
  {}

\lstnewenvironment{PVSNums}%
  {
  \if@firstcolumn
    \lstset{language=pvs, numbers=left, firstnumber=auto}
  \else
    \lstset{language=pvs, numbers=right, firstnumber=auto}
  \fi
  }
  {}

\lstnewenvironment{PVSNumsRight}%
  {
    \lstset{language=pvs, numbers=right, firstnumber=auto}
  }
  {}

\newcommand{\linefigpvs}[9]{

}

\lstdefinelanguage{pvsproof}{
  basicstyle=\tt \figuresize,
  mathescape=true,
  tabsize=4,
  sensitive=false,
  columns=fullflexible,
  keepspaces=false,
  flexiblecolumns=true,
  basewidth=0.05em,
}

\begin{document}
%%%
\title{Differentially Private Distributed Optimization}
% \numberofauthors{4}
\author{
%\alignauthor
Zhenqi Huang \ Sayan Mitra \ Nitin Vaidya \\
\{zhuang25, mitras, nhv\}@illinois.edu\\
Coordinate Science Laboratory\\
 University of Illinois at Urbana Champaign\\
   Urbana, IL 61801 
}

\maketitle
\begin{abstract}
In distributed optimization and iterative consensus literature, a standard problem is for $N$ agents to minimize a function $f$ over a subset of Euclidean space, where the cost function is expressed as a sum $\sum f_i$. In this paper, we study the private  distributed optimization (PDOP) problem with the additional requirement that the cost function of the individual agents should remain differentially private.
The adversary attempts to infer information about the private cost functions from the messages that the agents exchange. Achieving differential privacy requires that any change of an individual's cost function  only results in unsubstantial changes in the statistics of the messages. We propose a class of iterative algorithms for solving PDOP, which achieves differential privacy and convergence to the optimal value. Our analysis reveals the dependence of the achieved accuracy and the privacy levels  on the the parameters of the algorithm.  
We observe that to achieve $\epsilon$-differential privacy the accuracy of the algorithm has the order of $O(\frac{1}{\epsilon^2})$.
\end{abstract}

% !TEX root=main.tex
\section{Introduction}
\label{sec:intro}
We introduce the private distributed optimization problem (PDOP) in which $N$ agents are required to minimize a global cost function $f$ that is the sum $\Sigma_{i=1}^N f_i$ of $N$ cost functions for the individual agents. An instance of the problem arises when $N$ secretive agents (with their own convex travel costs) wish to agree on a rendezvous point in a country such that (a) the travel cost for the entire group is minimized and (b) an adversary reading all the communication between the agents is unable to deduce the cost functions for the individuals.  
%
%\hili{This problem can be solved in a centralized manner with a trusty entity who collects the cost functions of each agent and computes an optimal point with privacy-preserving algorithms. However, for networks without a trusty entity  or for a large-scale and time-varying network where centralized algorithm is ineffective,  we want to develop a distributed algorithm to the optimization problem. }
We study iterative distributed algorithms for solving this problem in which agents  exchange information about their current estimates for the  optimal point  and then update their estimates based on the information received from their neighbors. In doing so, however, the agents must preserve the privacy of their individual cost functions. The agents communicate over a  communication network in which the connectivity may change over time. While iterative solutions for distributed optimization have been explored  previously (see \cite{nedich09,nedich10,Tsitsiklis86}), to our knowledge this paper is the first attempt to achieve this goal while maintaining privacy.

An alternative to distributed iterative optimization is a centralized strategy
wherein a trusty {\em leader} is identified, with its task being to collect cost functions
from all the other agents, perform the optimization centrally, and the
distribute the results to all the other agents. 
While appealing for its simplicity, this strategy requires election of a leader, and maintenance
of routes from all agents to the leader.
The centralized scheme is then vulnerable to failure 
of the leader. Also, there is non-trivial cost of leader election
and route maintenance in time-varying topologies, and in some systems 
learning the network topology itself violates privacy
of the agents. Therefore, there has been significant interest in designing
completely distributed algorithms for network-wide optimization
and consensus. For instance, such algorithms have been designed for the
{\em smart grids}~\cite{dominguez2012sg} and {\em sensor networks}~\cite{consensus07}.

The notion of privacy we adopt is derived from  $\epsilon$-differential privacy~\cite{DiffPri:Dwork06,Dwork:2008:DPS:1791834.1791836,dwork2006our} applied to continuous bit streams in~\cite{Dwork10}.  This $\epsilon$-differential privacy ensures that an adversary with access to all the communication in the system---we call this an  observation sequence---cannot gain any significant information about the  cost function of any agent.

In~\cite{chaudhuri11,adam12private} the authors solve a privacy preserving optimization problem with two methods: output perturbation and objective perturbation. In this problem, the cost functions of the individual are assumed to have a template and the computation is done by an entity that have access to all the agents' individual data. In contrast, we study a class of problems that have to be solved in distributed ways without relying on any template of the individual cost functions.

In this paper, we propose a class of synchronous iterative distributed algorithms for solving PDOP.
Iterative algorithms proceed in round. In each round,  each agent participating in our algorithm executes three subroutines.
First, it adds a vector of  random noise, drawn from Laplace distribution, to its estimate for the optimal point and broadcasts this noisy estimate to its neighboring agents. Sharing noisy estimates enables the agent to protect the privacy of its cost functions. For convergence of the estimates to the optimal point, however, the noise added in successive rounds must decay down to 0. Indeed, in our algorithm the parameters of the successive Laplace distributions are chosen such that they  converge to the Dirac distribution. Next, the agent computes a weighted average over its neighbors' noisy broadcasts based on the communication graph of that round.
Finally, the agent computes a new estimate by moving the average value against the gradient of its own cost function according to a carefully chosen step size. 
%The step sizes are chosen so that they decay down to zero over rounds.

A key quantity which determines the amount of noise to be added in each round for achieving differential privacy is the sensitivity of the algorithm. Roughly, the sensitivity at round $t$ is the change in the observable behavior of the system at round $t$, namely the messages exchanged at round $t$, with change in the cost function of any agent (see Definition~\ref{def:sens}). 
For differential privacy, the ratio of the sensitivity and the parameter for the Laplace noise must be small (see Lemma~\ref{lem:lmc}). 
For the estimate of the optimal point to get arbitrarily close to the optimal point, standard iterative algorithms for distributed optimization (for example, the ones discussed in~\cite{nedich10}), require the sum of the step sizes to be infinite. 
This strategy, however, would increase  the sensitivity of the system for later rounds. That is, an adversary could begin to infer significant information about the individual cost functions. 
%
%Thus, unlike the standard algorithms, our algorithm uses step sizes that sum to a finite quantity.
Thus, unlike the standard algorithms and our previous algorithm for private consensus~\cite{wpes}, our algorithm for PDOP uses step sizes that sum to a finite quantity.
Assuming that the domain is bounded, we then establish convergence and both the level of differential privacy and the accuracy of the algorithm (Theorems~\ref{thm:dp} and~\ref{thm:acc}).

The algorithm has four parameters: the privacy level, the initial step size, the step size decay rate and the noise decay rate.  
%By choosing these parameters appropriately, the algorithm guarantees either $\epsilon$-differential privacy for arbitrary $\epsilon>0$ or $d$-accurate optimization for arbitrary $d>0$.  
%Our analysis reveals in detail the influence of these parameters on the accuracy level and the privacy level of the algorithm. 
Our analysis reveals that the accuracy level $d$ has the oder of the inverse-square of the privacy level  $\epsilon$.
% !TEX root=main.tex
\section{Preliminaries}
\label{sec:prelim}

The algorithms presented in this paper rely on random real numbers drawn according to the Laplace distribution.
For a constant $c>0$, $Lap(c)$ denotes the Laplace distribution with probability density function $p_c(x) \deq \frac{1}{2c}e^{-\frac{|x|}{c}}$. This distribution has mean zero and variance $2c^2$.  
For any $x,y\in\reals$, it can be shown that $\frac{p_c(x)}{p_c(y)}\leq e^{\frac{|y-x|}{c}}$.

For a natural number $N \in \naturals$, we denote the set $\{1, \ldots, N\}$ by $[N]$. 
For a  vector $v$ of length $n$,  the $i^{th}$ component is denoted by $v_i$. 
The transpose of $v$ is denoted by $v^T$.
For a vector $v$ in $\reals^n$ and $1\leq p\leq \infty$, $||v||_p$  stands for the standard $L^p$-norm for $v$. 
Without a subscript, $||\cdot||$ stands for $L^2$-norm.  That is, $||v||=\sqrt{v^Tv}$.  
For any vector $v\in\reals^n$, the inequality $||v||_2\leq ||v||_1\leq \sqrt n ||v||_2$ holds.

An {\em Euclidean projection} of a point $x\in \reals^n$ onto a set $\X\subseteq \reals^n$ is a point in $\X$ that is closest to 
$x$ measured by Euclidean norm. If there are multiple candidate points, one is chosen arbitrarily, and to reduce notational overhead we treat the Euclidean projection $\pj \X(x)$ as a function of $x$ and $\X$. 
That is, $y = \pj \X(x)$ if $y\in \X$ and $||y-x||\leq ||z-x||$ for any $z\in\reals^n \setminus \X$. 
A well known property of projection is that it does not increase the distance between points. That is,
$||\pj \X(x)-\pj \X(y)||\leq ||x-y||$ for any $x,y\in \reals^n$.

A differentiable function $f:\X\mapsto \reals$ is convex if for any $x,y\in \X$, $\grad f(x)^T(y-x)\leq f(y)-f(x)$. Moreover, if there exists a positive constant $c>0$ such that $\grad f(x)^T(y-x)\leq f(y)-f(x)-\frac{c}{2}||y-x||^2$, the function $f$ is said to be {\em strongly convex\/}. 
Strongly convex functions on compact domains have a unique minima~\cite{nesterov2007gradient}. 
For example, for constant $a\in \reals^n$, the quadratic  function $f(x) = ||x-a||^2$ is a strongly convex function in $\reals^n$, while the linear function $f(x) = x-a$ is convex but not strongly convex. 
%It can be checked that for any $x,y\in \reals^n$, the following holds: $2(x-a)^T(y-x)\leq ||y-a||^2-||x-a||^2-||y-x||^2$.

For a constant $\beta\in (0,1)$ and a convergent scalar sequence  $\{a_t\}_{t\in\naturals}$, the limit  $\lim_{t\rightarrow \infty} \sum_{s=1}^t \beta^{t-s}a_s$ exists.
\begin{prop}
\label{lem:convlim}
For a constant $\beta\in(0,1)$ and a convergent scalar sequence  $\{a_t\}_{t\in\naturals}$ such that $\lim_{t\rightarrow \infty}a_t=0$, the following holds:
\begin{equation}
\label{eq:convlim}
\lim_{t\rightarrow \infty} \sum_{s=1}^t \beta^{t-s}a_s = 0.
\end{equation}
\end{prop}
\noindent
Proof of Proposition~\ref{lem:convlim} can be found in appendix.
% !TEX root=main.tex
\section{The Private Distributed Optimization Problem}
\label{sec:problem}

A {\em Private Distributed Optimization\/} (PDOP) problem $\P$ for $N$ agents is specified by four parameters: 
\begin{enumerate}[(i)]
\item $\X\subseteq \reals^n$ is the domain of optimization, 
\label{it:pdopn}
\item $\F\subseteq \reals^\X$ is a set of real-valued, strongly convex and differentiable individual cost  functions on domain $\X$,
\item $f:\X\mapsto \reals$ is the global cost function which is a sum of $N$ cost functions in $\F$, that is, $f(x)\deq \sum f_i(x)$ with $f_i\in \F$ for each $i\in[N]$, and
\item $\A =\{A_t\}_{t\in\naturals}$ is a sequence of $N\times N$ matrices which specify the time-varying communication graph.
\end{enumerate}
More details on these parameters and additional assumptions we use for solving PDOP will be stated in Section~\ref{sec:dog}. 
In Section~\ref{sec:alg}, we introduce the class of algorithms we study in this paper. 
In Section~\ref{sec:dpca}, we formally state the requirements for solving PDOP. 

We describe the problem $\P$ as follows.
The system consists of $N$ agents. Each agent $i\in [N]$ is associated with an individual cost function $f_i: \X \mapsto \reals$.
% as the $i^{th}$ additive term of the global cost $f$. 
The individual cost $f_i$ is only known to agent $i$.
%The agents' communication is constrained by the sequence of matrices $\A$.
Together the agents aim to minimize:
\begin{eqnarray}
f(x) = \sum_{i\in [N]}f_i(x),
\end{eqnarray}
subject to the constraint $x\in \X$.
We define $f^*_\P\deq \min_{x\in \X}f(x)$ as the global minimum for $f$ 
and $x^*_\P\deq \argmin_{x\in \X}f(x)$ as the point in $\X$ that minimizes the cost function.  
For a PDOP $\P$ we denote its components and related quantities by $\X_\P, \F_\P, f_\P, \A_\P, f^*_\P$ and $x^*_\P$. 
We drop the subscript when it is clear from context. 
For a pair of PDOPs $\P$ and $\P'$, we will also denote the corresponding quantities by $\X, \F$, $\ldots$, and $\X'$, $\F'$, etc.  
To illustrate the idea, we present a private rendezvous problem in Example~\ref{ex:problem}.
%We omit the subscript $\P$ if it is clear from the context.
\begin{example}
\label{ex:problem}
\hili{
A number of $N$ agents live in a compact region $\X$ in a 2-D plane, where the address of each agent $i$ is a point $x_i\in \X \subseteq\reals^2$.   
The agents wants to decide an assembly point $x\in \reals^2$ without sharing their actual address. 
The cost of agent $i$ to go to point $x$ is the squared distance $f_i(x) = ||x_i-x||^2$.
Moreover, each agent can only keep in touch with a subset of the other agents.
Then, we can cast this problem as a PDOP, where
(i) the domain of optimization is $\X \subseteq \reals$, %is the convex hull covers the points $\{x_i\}_{i\in[N]}$,
(ii) the set of objective functions $\F=\{f| a\in\X, f(x)=||a-x||^2\}$, (iii) the global cost function $f(x)=\sum f_i(x)$,
and (iv) $\A$ is a sequence of matrices specify the possibly time-varying communication topology.
}
\end{example}

%\begin{example}
%\label{ex:ipo}
%\hili{
%As a company decides to issue its initial public offer (IPO),
%a number of $N$ buyers negotiate on their shares of IPO.
%Let $x_i\in[0,1]$ be the share buyer $i$ wants to buy.
%The {\em share distribution vector} is $x = [x_1,x_2,\dots,x_N] \in \X$, where $\X = \{x\in [0,1]^N | \sum_{i\in[N]}x_i  \leq 1\}$ is the convex domain of optimization.
%The cost of buyer $i$ is a convex function $f_i(x)$ of the share distribution vector $x$, which is a valuable information of the buyer.
%Then, the group of $N$ buyers communicate to decide an optimal share distribution vector $x$ without releasing the individual cost functions.
%}
%\end{example}

\subsection{Domain, Cost Function and Communication Graph}
\label{sec:dog}

We make the following assumptions on the domain of optimization and the set of individual cost functions throughout the paper.
\begin{assumption}[Convexity and compactness]
\label{ass:function}
         
\begin{enumerate}[(i)]
\item The set $\X$ is compact and convex. Let $C_1 \deq \sup_{x,y\in\X}||x-y||$ denote the diameter of $\X$.
\item The gradients of all the individual cost function are uniformly bounded.
That is, there exists $C_2>0$ such that for any $x\in \X$ and any $g\in\F$, $||\grad g(x)||\leq C_2$.
\item The functions in $\F$ are strongly convex. That is, there exists $C_3>0$ such that for any $x,y\in\X$ and for any 
$g\in \F$,  $\grad f(x)^T(y-x)\leq f(y)-f(x)-\frac{C_3}{2}||y-x||^2$. 
\end{enumerate}
\end{assumption}
The first part of Assumption~\ref{ass:function} is standard which guarantees the existence of an optimal solution.
The second part  holds if the magnitude of the gradients of individual cost functions do not grow unbounded.
In many optimization problems, it is standard to assume the cost function to be convex, for which gradient based method is effective. Strong convexity provides a stronger bound on the gradient term $\grad f(x)^T(y-x)$ which is necessary of analyzing the accuracy of our algorithm.  It can be checked that the PDOP introduced in Example~\ref{ex:problem} satisfies Assumption~\ref{ass:function}.
%Let $F=\F^{N}$ denote the set of all vectors of $N$ individual cost functions.
%Fix a domain $\X$, a DOP is uniquely defined by fixing any vector $f\in F$. 
%For a DOP $f$, we denote its $i^{th}$ component as $f_i$. We define $OPT_f\deq \min_{x\in \X}\sum_{i\in [N]}f_i(x)$ as the global minimum valuation of $f$. $x^*_f=\argmin_{x\in \X}\sum_{i\in [N]}f_i(x)$ denotes the set of arguments minimize the object.  We omit the subscript if it is clear in the context.

We assume a synchronous model of distributed computation though the communication network among the agents is time varying. 
%system, as is the norm in the related literature~\cite{Nedich-paper}.
We model the communication network at round $t$ as a weighted graph $\G_t=(\V,\E_t,\W_t)$, where 
\begin{inparaenum}[(i)]
\item $\V=[N]$ is the set of agents,
\item $\E_t \subseteq \V\times \V $ is the set of edges over which information is exchanged at round $t\in \naturals$, and 
\item $\W_t:\E_t\mapsto (0,1]$ is the weighted function labels each edge with a positive weight.
\end{inparaenum}
The graph $\G_t$ is represented by an $N\times N$ matrix $A_t$, where
the entry $a_{i,j}(t) \deq \W_t(i,j)$ if $(i,j)\in \E_t$ otherwise $a_{i,j}(t):=0$.
We assume that the matrix $A_t$ is doubly stochastic. 
That is, for each $i\in[N]$, $\sum_{j\in[N]}a_{i,j}(t)  = 1$ and for each $j\in[N]$, $\sum_{i\in[N]}a_{i,j}(t)=1$. 
Roughly, a doubly stochastic $A_t$ ensures that each agent's decision has an equal influence on the final decision. This statement will become clearer after we introduce the algorithm. 
There are existing  distributed algorithms to derive a  doubly stochastic matrix among a network (see e.g.~\cite{dominguezmatrix}).
We use the following technical assumptions of the time-varying communication network $\A$ throughout the paper.
\begin{assumption}[Robust connectivity]
\label{ass:graphrobust}
We assume that for each $t\in\naturals$, the graph $A_t$ is strongly connected. In addition, there exists a minimal connection strength $\eta\in(0,1]$ 
such that for each $t\in\naturals$:
\begin{enumerate}[(i)]
\item $a_{i,i}(t) \geq \eta$ for each $i\in[N]$. And
\item $a_{i,j}(t) > 0$ implies that $a_{i,j}(t)\geq \eta$.
\end{enumerate}
\end{assumption}
This assumption guarantees that there exists a path in the graph linking each pair of the agents and the sum of weights along the path is lower bounded. 
%Roughly, it lower bounds the influence of one agent to any other agent.

\subsection{Iterative Distributed Algorithms for PDOP}
\label{sec:alg}
We study an class of iterative distributed algorithms for solving PDOP. 
As shown in Algorithm~\ref{alg:dop}, $R$, $U$ and $F$ are functions or subroutines, 
which when instantiated will give candidate algorithms. %will be defined later; 
The constant $T$ is the total number of rounds over which the algorithm is executed and it determines the accuracy of the final answer. 
The agents have internal states. 
An agent's state is defined by the valuations of individual variables. 
Each agent has four internal variables, which are
\begin{inparaenum}[(i)]
\item $x_i\in\X$ is agent $i$'s current estimate of the optimal point; it is initialized to an arbitrary point $x_{i0}$ in $\X$,
\item $y_i\in \reals^n$ is the value agent $i$ broadcasts to other agents,
\item $z_i\in \reals^n$ is the value agent $i$ computes based on  the values it receives from its 	neighbors,
\item $t\in \naturals$ is the current round number, and
\item $\mathit{buffer}$ is an ordered set which stores the messages received by agent $i$ in a given round from its neighbors. 
\end{inparaenum}

\begin{algorithm}[h]                      
%\fontsize{\footnotesize}
\caption{Template for iterative solution of PDOP.}          
%\begin{multicols}{2}
\label{alg:dop}                           
\begin{algorithmic}[1]	                    
%\KwIn{ $\Delta$}
\STATE \bf{Input:} $f_i,\X,\A$
\STATE $x_i \gets x_{i0}$\;  
\FOR{$t = 1: T$}   
	\STATE $y_i \gets R(x_i,t)$\;
	\STATE Broadcast$(y_i)$\;
	\STATE $\mathit{buffer}_i \gets$ Receive$()$\;
	\STATE $z_i\gets F(A_t,\mathit{buffer}_i)$ \;
          \STATE $x_i \gets U(z_i,t,f_i,\X)$\;
\ENDFOR 
\RETURN $R$
\end{algorithmic}
%\end{multicols}
\end{algorithm}

Message exchanging between agents is assumed to be atomic.
That is, the {\bf Receive}($y_i$) routine of agent $i$ broadcasts $y_i$ to all his neighbors and the {\bf Receive}() routine receives all neighbors' broadcasts immediately.
This can be implemented by underlying message exchanging protocols.
%The optimization algorithm presented above is specified by three (possibly randomized) subroutines:
%\begin{inparaenum}[(i)]
%\item a report subroutine $R$, 
%\item a fusion subroutine $F$, and
%\item an update subroutine $U$.
%\end{inparaenum} 
In each round $t\in\naturals$, the algorithm has four phases: 
\begin{inparaenum}[(i)]
\item each agent executes a subroutine $R$ to compute the value to report ($y_i$) based on his individual value ($x_i$) (line 4),
\item each agent broadcasts its value ($y_i$) and receives all neighbors' reports (line 5-6),
\item each agent executes a subroutine $F$ to compute a aggregate value ($z_i$) based on its neighbors' messages (line 7), and
\item each agent executes a subroutine $U$ to compute a new individual value ($x_i$) that reduces the individual cost function $f_i$ (line 8).
\end{inparaenum}
%An iterative distributed algorithm is specified if the subroutines $R$, $F$ and $U$ are defined.
%
%At time $t\in \naturals$, the agent $i$ executes the three subroutines as following:
%\begin{eqnarray}
%\label{eq:ri}
%y_i(t) &=& R_i(x_i(t-1),t)\\
%\label{eq:fi}
%z_i(t) &=& F_i(y(t),A_t)\\
%\label{eq:ui}
%x_i(t) &=& U_i(x_i(t),z_i(t),t) 
%\end{eqnarray}
%

We denote $x_i(t)$ as the valuation of  $x_i$ at the end of round $t$.
We denote the aggregate state $x(t)$ as a vector of the $N$ individual valuations: $x(t)\deq [x_1(t),\dots,x_N(t)]$.
$y_i(t)$, $z_i(t)$,  $y(t)$ and $z(t)$ are similarly defined as valuations of individual variables and vectors  at the end of round $t$.
Each round of  a iterative distributed algorithm transforms the state vector of the entire system to a new state vector. 
%each individual state to another.
An {\em execution\/} of such an algorithm for a given PDOP, is a possibly infinite sequence of the form 
$\alpha=x(0),\langle x(1),y(1),z(1),\mathit{buffer}(1)\rangle,$ $\langle x(2),y(2),z(2),$ $\mathit{buffer}(2)\rangle, \dots$. 
The observable part of such an execution are the corresponding infinite sequence of messages $y(1), y(2), \ldots$.
We denote the observation mapping $\R(\alpha)\deq y(1),y(2),\dots$ which gives the sequence of messages exchanged for the execution $\alpha$..

Note that the set of messages stored in $\mathit{buffer}_i(t)$ is uniquely specified by the vector $y(t)$ and the communication graph for the round $A_t$. 
Thus, for deterministic subroutines $R$, $F$, and $U$, and particular choices of the initial valuations of the variables, and a given PDOP $\P$ an iterative distributed algorithm has a unique execution. 
For fixed (possibly randomized) subroutines $U, R, F$, and a fixed initial state $x(0)$, let $\mathit{Obs}$ denote the set of all sequences of messages that the resulting algorithm can produce for any PDOP problem\footnote{Here we are suppressing the dependence of $Obs$ on $U, R, F$ and $x(0)$ for notational convenience.}. 

In this paper, we will study randomized versions of Algorithm~\ref{alg:dop}.
%Some times we omit the sequence of $\mathit{buffer}(t)$ when studying an execution.
For a fixed choice of these randomized subroutines (to be stated in Section~\ref{sec:alg}), $\Xi_\P$ denotes the set of all executions of the resulting algorithm for a given PDOP $\P$ and a given set of initial conditions\footnote{Here we are suppressing the dependence of $\Xi_\P$ and $\pr_\P$ on $U, R, F$ and $x(0)$ for notational convenience.}. 
The probability measure over the space of executions $\pr_\P$ is defined in the standard way by first defining a $\sigma$-algebra of cones over the space of executions, and then by defining the probability of the cones by integrating over $\mu$ (see for example~\cite{segala1996modeling,Mitra07PhD}).   

\subsection{Convergence, Accuracy and Differential Privacy}
\label{sec:dpca}

An iterative distributed algorithm solves the PDOP problem if the estimates of all the agents converge to a common value and the algorithm preserves differential privacy of the $f_i$'s. Furthermore, we want this convergence point close to the optima $x^*$ of $f$. 
\begin{definition}[Convergence]
\label{def:con}
An iterative distributed algorithm converges  if for any PDOP $\P$ and any initial configuration, for any agents $i,j\in [N]$,
\[
\lim_{t\rightarrow \infty} \ep||(x_i(t)-x_j(t)||=0, 
\]
where the expectation is taken over the coin-flips of the algorithm, that is, the randomization in the $R$, $F$ and $U$ subroutines of the individual agents.
\end{definition}

%In this paper, we discuss a notion of geometrically convergence. That is, the expected distance between two agents is bounded by a geometric sequence with ratio $r\in(0,1)$.
%%the algorithm converges geometrically, we can further discuss how fast it converges.
%\begin{definition}[Convergence]
%\label{def:con}
%For $r\in (0,1)$, an iterative distributed algorithm is  $r$-convergent if for any PDOP $\P$ and any initial configuration,
%for any time $t\in \naturals$, for any agents $i,j\in [N]$, there exists a constant $r\in(0,1)$ and a constant $\theta>0$ such that
%\[
%\ep||(x_i(t)-x_j(t)||\leq \theta r^t, 
%\]
%where the expectation is taken over the coin-flips of the algorithm, that is, the randomization in the $R$, $F$ and $U$ subroutines of the individual agents.
%\end{definition}
%\noindent
%The smaller $r$ is, the faster the algorithm converges;  the closer $r$ is to 1, the slower the convergence. Also, it can be shown that for any $r\in(0,1)$, if an algorithm is $r$-convergent, then $\lim_{t\rightarrow \infty} \ep||(x_i(t)-x_j(t)||=0$. That is, roughly, the local values will eventually match up.

We define $\bar x(t) \deq \frac{1}{N}\sum_{i\in[N]}x_i(t)$ as the average of the  individual agent estimates at the end of round $t$.  We define the accuracy of the algorithm by the expected squared distance of the average to the optima $x^*$.
\begin{definition}[Accuracy]
\label{def:acc}
For a $d\geq 0$, an iterative distributed algorithm is said to be $d$-accurate if,
\begin{equation}
\label{eq:defacc}
%\limsup_{t\rightarrow \infty}\ep[f(\bar x(t))]  \leq f^*_\P+ d, 
\lim_{t\rightarrow \infty}\ep||\bar x(t) - x^*||^2\leq d,
%\footnote{The  limit superior  is used to guarantee that Equation~\eqref{eq:defacc} is well-defined even for algorithm that does not converge.  }
\end{equation}
where the expectation is taken over the coin-flips of the algorithm.
\end{definition}
The smaller the $d$-accuracy, the more accurate the algorithm.
If the algorithm  converges to the exact global optimal point $x^*$, then it is $0$-accurate.

Our definition of privacy is a modification of the notion of {\em differential privacy\/} introduced in~\cite{Dwork10} in the context of streaming algorithms.
We consider an adversary with full access to all the communication channels. 
That is, he can peek inside all the messages ($y(t)$) going back and forth between the agents. 
%Furthermore, the intruder can access the internal state of a subset of the agents. 
%
%
%Roughly, a randomized algorithm for the agents solves the synchronous optimization problem if eventually all the agents converge to the glable optimum with high probability and it guarantees that the intruder cannot learn about the initial private client values with any high level of confidence. We proceed to precisely define accuracy, convergence, and privacy.
%  
%
%%%%%%%%%%%%%%%%%%%%%%%
%Formally, for an iterative distributed algorithm,  there is an inverse observation mapping $\R^{-1}$ that maps from 
%\begin{inparaenum}[(i)]
%\item a PDOP $\P$,
%\item an observation sequence $\rho\in\mathit{Obs}$, and
%\item an initial state vector $x(0)$
%\end{inparaenum}
%to the set of corresponding  executions $\{\alpha \in \Xi_\P : \R(\alpha) = \rho \wedge \alpha(0)=x(0) \}$.
%
For a give PDOP $\P$, observation sequence of messages $\rho \in \mathit{Obs}$, and an initial state $x(0)$, then 
$\R^{-1}(\P,\rho,x(0))$ is the set of executions $\{\alpha \in \Xi_\P : \R(\alpha) = \rho \wedge \alpha(0)=x(0) \}$ that can generate the observation $\rho$. 
\begin{definition}[Adjacency]
\label{def:adj}
Two PDOPs $\P$ and $\P'$ are {\em adjacent},  if the following holds: 
\begin{enumerate}[(i)]
\item $\X=\X'$, $\F=\F'$ and $\A=\A'$, that is, the domain of optimization, the set of individual cost functions and the communication graphs are identical, and 
\item there exists an $i \in [N]$, such that $f_i \neq f'_i$ and for all $j \neq i$, $f_j = f'_j$. 
\end{enumerate}
\end{definition}
%
%\begin{definition}[Differential Privacy]
%\label{def:dp}
%Let $f \subset \reals^N$ be the domain of global state. 
%%equipped with metric $m(\cdot,\cdot)$.  
%A randomized algorithm  preserves {\em $\epsilon$-differential privacy\/}  if for any initial states $\Theta \in \X^N$ and $Y' \subseteq Y$, and 
%for all pairs of $\delta$-{\em adjacent\/} DOP $f, f'\in \F^N$ 
%\[
%Pr[{\bf Alg}(f,\Theta) \in Y'] \leq e^{\epsilon \delta} Pr[{\bf Alg}(f',\Theta) \in  Y'].
%\]
%\end{definition}
%This definition of adjacency uses a 1-norm whereas the standard definition (found in~\cite{Dwork06}, for example) uses the Hamming distance. This choice of the metric has ramifications on the privacy guarantees. In cases where each agent's individual value comes from a bounded set, by letting $\delta$ equals to the range of individual value,  Definition~\ref{def:pri0} subsumes the standard definition. In cases where each agent's individual value comes from an unbounded set, the sensitivity of a query can be unbounded. In such cases, the algorithms introduced in this paper fail to provide differential privacy (in the sense of~\cite{Dwork06}) and the 
%$\delta$-adjacency notion becomes useful.
That is, two PDOP are adjacent if only one agent changed its individual cost function while all other parameters are identical.
\begin{definition}
\label{def:dp}
For an $\epsilon\geq 0$, the iterative distributed algorithm  is $\epsilon$-{\em differentially private\/}, if for any two adjacent PDOPs $\P$ and $\P'$, any set of observation sequences $Y\subseteq \mathit{Obs}$ and any initial state $x(0)\in \X^N$
\begin{align}
\label{eq:dpdef}
\pr[\R^{-1}(\P,Y,x(0))] \leq e^\epsilon \pr[\R^{-1}(\P',Y,x(0))],
\end{align}
%The probability is from the coin toss of the algorithm.
where the expectation is taken over the coin-flips of the algorithm.
\end{definition}
Roughly, the notion of $\epsilon$-differential privacy ensures that an adversary with access to all the  observation sequence cannot gain information about the individual cost function of any agent with any significant probability. The probability measure is over the $\sigma$-algebra of of cones over the set of executions. 
%This technique helps us to rigorously reason the probability of possibly infinite executions.
A smaller $\epsilon$ suggests a higher privacy level.
In the rest of the paper, we discuss algorithms to solve PDOP which guarantee $\epsilon$-differential privacy and $d$-accuracy.
%We call this problem {\em private distributed optimization problem} (PDOP).
%!TEX root=main.tex
\section{An Algorithm for PDOP instantiating the Subroutines}
\label{sec:mechanism}
\subsection{Algorithm Description}
\label{sec:alg}
In Section~\ref{sec:alg}, we introduced a class of iterative distributed algorithms in terms of the subroutines $R$, $F$ and $U$.
In this section, we instantiate the subroutines and analyze the convergence, accuracy and differential privacy of the resulting algorithm.  The algorithm has 4 parameters: (i) the privacy parameter $\epsilon>0$, (ii) the initial step size parameter $c>0$, 
(iii) the step size decay rate $q\in(0,1)$ and (iv) the noise decay rate $p\in(q,1)$.

The subroutine $R$, shown in Algorithm~\ref{alg:r}, computes a value to broadcast based on agent $i$'s local value $x_i$ at  round $t$.  The agent first generates a vector of $n$ (which is the length of $x_i$) noise values drawn  independently from the Laplace distribution $Lap(M_t)$ with parameter $M_t$, which we will define later. 
%Lemma~\ref{lem:lmc} shows how to select an $M_t$. 
The value to report is the sum of $x_i$ and the noise vector $w_i$. 

\begin{algorithm}[h]                      
\caption{Subroutine $R$}          
%\begin{multicols}{2}
\label{alg:r}                           
\begin{algorithmic}[1]	                    
%\KwIn{ $\Delta$}
\STATE {\bf Input:} $x_i\in \X, t\in \naturals$
\STATE $w_i \sim Lap(M_t)$\;      
\STATE $y_i = x_i + w_i$\;
\RETURN $y_i$
\end{algorithmic}
%\end{multicols}
\end{algorithm}

The subroutine $F$, shown in Algorithm~\ref{alg:f}, computes a value $z_i$ based on all the neighbors'  messages. 
In this subroutine, the agent $i$ first read its neighbors' broadcasts from its own buffer.
Recall that $a_{i,j}(t)$ is the entry on the $i^{th}$ column and $j^{th}$ row of  a doubly stochastic matrix $A_t$.
Thus, the value $z_i= \sum_{j\in [N]} a_{ij}(t)y_j$ is indeed the weighted average of neighbors' broadcast based on the communication graph $A_t$.

\begin{algorithm}[h]                      
\caption{Subroutine $F$}          
%\begin{multicols}{2}
\label{alg:f}                           
\begin{algorithmic}[1]	                    
%\KwIn{ $\Delta$}
\STATE {\bf Input:} $A_t, \mathit{buffer}_i$
\STATE For all $j\in[N]$, $y_j\gets read(\mathit{buffer}_i,j)$\;
\STATE $z_i= \sum_{j\in [N]} a_{ij}(t)y_j$\;      
\RETURN $z_i$
\end{algorithmic}
%\end{multicols}
\end{algorithm}

The subroutine $U$ is shown in Algorithm~\ref{alg:u}. 
In this subroutine, the agent computes a new local value $x_i$ by moving from $z_i$ against the gradient of $f_i$. 
Roughly, this computation reduces the individual cost function $f_i$ from the point $z_i$. 
The parameter $\gamma_t$ is the step size at round $t$, which we will define later in Equation~\eqref{eq:para}. The projection $\pj \X$ guarantees that the estimate for agent $i$ is in $\X$.

\begin{algorithm}[h]                      
\caption{Subroutine $U$}          
%\begin{multicols}{2}
\label{alg:u}                           
\begin{algorithmic}[1]	                    
%\KwIn{ $\Delta$}
\STATE {\bf Input:} $z_i\in \X, t \in\naturals,f_i, \X$   
\STATE $x_i \gets \pj \X [z_i-\gamma_t(\grad f_i(z_i))]$\;
\RETURN $x_i$
\end{algorithmic}
%\end{multicols}
\end{algorithm}

From Algorithm~\ref{alg:r}-\ref{alg:u}, we can write down the computation of each agent $i$ at round $t$ as following three equations: 
\begin{eqnarray}
\label{eq:mech1}
y_i(t) &=& x_i(t-1)+w_i(t) \\
\label{eq:mech2}
z_i(t) &=& \sum_{j\in [N]} a_{ij}(t)y_j(t). \\
\label{eq:mech3} 
x_i(t) &=& \pj \X[ z_i(t)-\gamma_t(\grad f_i(z_i(t)))] .
\end{eqnarray}
In Equations~\eqref{eq:mech1}-\eqref{eq:mech3}, there are two undefined parameters: the noise parameter $M_t$ and the step size $\gamma_t$ at round $t$. 
We propose to choose $M_t$ and $\gamma_t$ as geometrically decaying sequences depends on 4 parameters ($c,q,p,\epsilon$):
\begin{equation}
\label{eq:para}
\begin{array}{rl}
\gamma_t&=cq^{t-1}\\
M_t &=2C_2\sqrt n \frac{c p}{\epsilon(p-q)} p^{t-1}. 
\end{array} 
\end{equation}
where $c>0$ is the initial step size parameter, $q\in(0,1)$ is the step size decay parameter, $p\in(q,1)$ is the noise decay parameter and $\epsilon>0$ is the privacy parameter.
The initial noise ($2C_2\sqrt n \frac{c}{\epsilon(p-q)}$)  depends on the three parameters $c,q,p,\epsilon$, the dimension of domain $n$ (See~\ref{it:pdopn} in Definition of PDOP) and constant $C_2$ from Assumption~\ref{ass:function}. 
%Moreover, the step size $\alpha_t$  decays with rate $q/2$, two times faster than the noise parameter $M_t$.
Note that the noise distribution converges to Dirac distribution and the step size converges to 0. 

Thus, the algorithm we introduced to solve the PDOP (Algorithm~\ref{alg:r}-\ref{alg:u}) has three tunable parameters: the initial noise parameter $c$, its decaying rate $q$ and step size decay rate $p$.
Later we show that by any choice of $c>0$, $q\in(0,1)$ and $p\in(q,1)$, the iterative distributed algorithm is $\epsilon$-differentially private, convergent, and ensures certain level of accuracy.
In Section~\ref{sec:acc} we will discuss a tradeoff between convergence rate and accuracy.

\subsection{Differential Privacy}
\label{sec:dp}

Recall in Definition~\ref{def:adj}, two PDOP $\P$ and $\P'$ are adjacent if only one term of $f'$  is different from that of $f$. 
The notion of sensitivity of a mechanism captures the maximum change in the states ($x(t)$) for two adjacent PDOPs.
Recall that in Section~\ref{sec:dpca}, we introduced a inverse observation mapping $\R^{-1}(\P,\rho,x(0))$ that maps a PDOP $\P$, an observation sequence $\rho$ and an initial configuration $x(0)$ to a set of executions, each of which is in the form $\alpha=x(0),\langle x(1),y(1),z(1),$ $\mathit{buffer}(1),\rangle, \langle x(2),y(2),z(2),\mathit{buffer}(2)\rangle,\dots$. Let $\R_{x(t)}^{-1}(\P,\rho,x(0))$ be the set of $x(t)$ component from each of the executions $\alpha$ in the set the set of the executions $\alpha \in \R^{-1}(\P,\rho,x(0))$.
\begin{definition}
\label{def:sens}
At each round $t\in\naturals$, for any observation sequence $\rho\in \mathit{Obs}$, any initial state $x(0) \in \X^N$ and any adjacent PDOPs $\P,\P'$,
We define the sensitivity of Algorithm~\ref{alg:dop} as
\[
\Delta(t) \deq \sup_{ \substack{x\in\R_{x(t)}^{-1}(\P,\rho,x(0))\\ x'\in\R_{x(t)}^{-1}(\P',\rho,x(0))}} ||x-x'||_1,
\]
where the norm used is $L^1$-norm. 
\end{definition}
We will show that $\Delta(t)$ is bounded for any $t \in \naturals$ for the algorithm. We state the following lemma which is a sufficient condition on the amount of noise to guarantees $\epsilon$-differential privacy.
\begin{lemma}
\label{lem:lmc}
At each round $t\in \naturals$, if each agent adds a noise vector $\omega_i(t)$  consisting $n$ Laplace noise independently drawn from $Lap(M_t)$ such that $\sum_{t=1}^\infty\frac{\Delta(t)}{M_t}\leq \epsilon$,
then the iterative distributed algorithm is $\epsilon$-differentially private.
\end{lemma}
\begin{proof}
Fix any pair of adjacent PDOP $\P$ and $\P'$, any set of observation sequence $Obs$ and any initial state $x(0)\in \X$.
For simplicity, we denote the sets of executions $\R^{-1}(\P,Obs,x(0))$ and $\R^{-1}(\P',Obs,x(0))$ by $A$ and $A'$ respectively.
First we introduce a proposition of the uniqueness of the mapping $\R^{-1}$.  The proof of Proposition~\ref{lem:uniqueinverse} can be found in appendix.
\begin{prop}
\label{lem:uniqueinverse}
For any PDOP $\P$, any observation sequence $\rho\in Obs$, for any initial state $x(0)\in \X^N$, $\R^{-1}(\P,\rho,x(0))$ is a singleton set.
\end{prop}
We define a correspondence $B$ between the sets $A$ and $A'$. 
For $\alpha \in A$ and $\alpha' \in A'$, $B(\alpha)=\alpha'$ if and only if they have the same observation sequence. That is $\R(\alpha)= \R(\alpha')$.
Fix any observation sequence $\rho$ in $Obs$, there is an unique execution $\alpha \in A$ that can produce the observation. Similarly, $\alpha'$ is also unique in $A'$.
So $B$ is indeed a bijection.
we relate the probability measures of the sets of executions $A$ and $A'$. 
\begin{equation}
\label{eq:l1dp}
\frac{\pr[\R^{-1}(f,Obs,x(0))]}{\pr[\R^{-1}(f',Obs,x(0))]} = \frac{\int_{\alpha\in A}\pr[\alpha] d \mu}{\int_{\alpha'\in A'}\pr[\alpha' ] d \mu'}.
\end{equation}
Changing the variable using the bijection $B$ we have,
\begin{equation}
\label{eq:l1cv}
\int_{\alpha'\in A'}\pr[\alpha' ] d \mu' = \int_{B(\alpha)\in A'}\pr[B(\alpha) ] d \mu = \int_{\alpha\in A}\pr[B(\alpha) ] d \mu
\end{equation}
From Algorithm~\ref{alg:r}-\ref{alg:u}, recall that we fixed the  observation sequence $\rho$, the probability comes from the noise $w_i(t)$. 
%That is,
%\[
%\int_{\alpha\in A}\pr[\alpha ] d \mu 
%= \int_{\alpha\in A}\pr[\xi |  \rho] d \mu.
%\] 
%where $\xi$ is the sequence of state vector $x(t)$ corresponding to $\alpha$; $\rho$ is the corresponding observation sequence. 
Along the sequence $\xi$, $x_i(t)$ is a vector of length $n$.  We denote the $k$ state component of $x_i(t)$ by $x_i^{(k)}(t)$. From Algorithm~\ref{alg:r},  $y_i(t)$ is obtained by adding $n$ independent noise to $x(t)$, from the distribution $Lap(M_t)$, it follows that the probability density of an execution $\alpha$ is reduced to
\begin{equation}
\label{eq:l1pdf}
%\pr[\xi |  \rho] 
\int_{\alpha\in A}\pr[\alpha ] d \mu= \prod_{ \substack{i\in [N], k\in [n] \\ t\in\naturals}}
p_{M_t}(y_i^{(k)}(t)- x_i^{(k)}(t)),
\end{equation}
where $p_b(x)$ is the probability density function of $Lap(b)$ at $x$. 
Then, we relate the distance at time $t$ between the state of $\alpha$ and $B(\alpha)$ with the sensitivity $\Delta(t)$. 
By the Definition~\ref{def:sens}, we have
\[
|| x(t)-x'(t)||_1 \leq\Delta(t).
\]
The norm in above equation is $L^1$-norm. The global state $x(t)$ consists of $N$ local state $x_i(t)$, each of which has $n$ component. 
So $( x(t)-x'(t))$ lives in space $\reals^{nN}$. By definition of $L^1$-norm:
\[ 
\begin{array}{rl}
\sum_{i=1}^N\sum_{k=1}^n | x_i^{(k)}(t)-x_i'^{(k)}(t)| = || x_i(t)-x'_i(t)||_1 \leq \Delta(t).
\end{array}
\]
Recall that by definition of $B$, the observations of $\alpha$ and $B(\alpha)$ match, that is $ y(t) = y'(t)$. 
From the property of Laplace distribution introduced in Section~\ref{sec:prelim},
\begin{equation}
\label{eq:l1lap}
\begin{split}
&\prod_{i\in[N] k\in[n]} \frac{p_{M_t}(y_i^{(k)}(t)- x_i^{(k)}(t))}{p_{M_t}(y_i'^{(k)}(t)- x_i'^{(k)}(t))} \\
\leq& \prod_{i\in[N],k\in[n]} exp\left({\frac{|y_i^{(k)}(t)- x_i^{(k)}(t) - y_i'^{(k)}(t)+ x_i'^{(k)}(t)|}{M_t}} \right) \\
=& \prod_{i\in[N],k\in[n]} exp \left({\frac{|x_i^{(k)}(t) - x_i'^{(k)}(t)|}{M_t}} \right) \\
=& exp\left({\sum_{i\in[N],k\in[n]}\frac{|x(\alpha(t)) - x(B(\alpha)(t))|}{M_t} }\right) 
\leq  e^{\frac{\Delta(t)}{M_t}}. \\
\end{split}
\end{equation}
Combining Equation~\eqref{eq:l1dp}, \eqref{eq:l1cv}, \eqref{eq:l1pdf} and \eqref{eq:l1lap}, we derive
\begin{equation*}
\begin{split}
&\frac{\pr[\R^{-1}(f,Obs,x(0))]}{\pr[\R^{-1}(f',Obs,x(0))]} 
=  \frac{\int_{\alpha\in A}\pr[\alpha] d \mu}{\int_{\alpha\in A}\pr[B(\alpha)] d \mu} 
\leq  \prod_{t\in \naturals} e^{\frac{\Delta(t)}{M_t}} \\
\leq& e^{\sum_{t\in\naturals}\frac{\Delta(t)}{M_t}}
\end{split}
\end{equation*}
If $M_t$ satisfy $\sum_{t=0}^\infty\frac{\Delta_D(t)}{M_t}\leq \epsilon$, then $\prod_{t\in\naturals}e^{\frac{\Delta(t)}{M_t}} \leq e^\epsilon$.
Thus the lemma follows.
\end{proof}

Lemma~\ref{lem:lmc} states that by adding Laplace noises drawn independently from some Laplace distribution, the iterative distributed algorithm defined by Algorithm~\ref{alg:r}-\ref{alg:u} guarantees $\epsilon$-differential privacy.
The parameters of the noise to add depends on the sensitivity of the algorithm. 
In the next lemma, we state a bound of the sensitivity of our proposed algorithm.
\begin{lemma}
\label{lem:sensitivity}
Under Assumption~\ref{ass:function}, 
the sensitivity of the proposed algorithm is
\[
\Delta(t) =2C_2\sqrt n \gamma_t.
\]
\end{lemma}  
\begin{proof}
Fix any observation sequence $\rho$, any initial state $x(0)\in\X^N$ and any adjacent $\P,\P'$.  
Let 
$\R^{-1}(\P,\rho,x(0)) = x(0), \langle x(1),$ $y(1),z(1),buffer(1) \rangle,\dots$ 
and 
$\R^{-1}(\P',\rho,x(0))=x'(0),\langle x'(1),$ $y'(1),z'(1),buffer'(1)\rangle,\dots$ be the executions for PDOP $\P$ and $\P'$ respectively.

By fixing the observation sequence $\rho$ for both executions,  we have $y(t)=y'(t)$ for all $t$. 
From Algorithm~\ref{alg:f},  $z_i(t) = \sum_{j\in [N]} a_{ij}(t)y_j(t) = \sum_{j\in [N]} a_{ij}(t)y'_j(t) = z'_i(t)$ for each $i\in[N]$ and each round $t$.
From Definition~\ref{def:adj},   $f$ and $f'$ are identical except for the $i^{th}$ components.
Thus, by applying Algorithm~\ref{alg:u}, we have:
\begin{eqnarray*}
&&||\R_{x(t)}^{-1}(\P,\rho,x(0))-\R_{x(t)}^{-1}(\P',\rho,x(0))||_1 \\
& =& ||z_i(t)-\gamma_t(\grad f_i(z_i(t))) - z'_i(t)+\gamma_t(\grad f'_i(z'_i(t)))||_1\\
& =&\gamma_t||\grad f_i(z_i(t)))-\grad f'_i(z'_i(t)))||_1
\end{eqnarray*}
From Assumption~\ref{ass:function}, the $L^2$ norm $||\grad f_i(z_i(t)))-\grad f'_i(z_i(t)))||\leq 2C_2$. 
By the norm inequality introduced in Section~\ref{sec:prelim}, we have,
\begin{eqnarray*}
\Delta(t) &=&\sup_{ \substack{x\in\R_{x(t)}^{-1}(\P,\rho,x(0))\\ x'\in\R_{x(t)}^{-1}(\P',\rho,x(0))}} \gamma_t||\grad f_i(z_i(t)))-\grad f'_i(z_i(t)))||_1\\
&\leq& 2C_2\sqrt n \gamma_t. 
\end{eqnarray*}
\end{proof}
With Lemma~\ref{lem:lmc} and~\ref{lem:sensitivity}, it directly follows that our algorithm guarantees $\epsilon$-differential privacy.

\begin{theorem}
\label{thm:dp}
The proposed algorithm (Algorithm~\ref{alg:dop}-\ref{alg:u}) guarantees $\epsilon$-differential privacy with any choice of $c>0$, $q\in(0,1)$ and $p\in(q,1)$. 
%with $\epsilon=\frac{2C_2\sqrt nc_2q_1}{c_1(q_1-q_2)}$.
\end{theorem}
\begin{proof}
Recall that in Equation~\eqref{eq:para}, the step size at round $t$ is $\gamma_t=cq^{t-1}$. Besides  the Laplace noise at round $t$ is drawn from distribution $Lap(M_t)$ with $M_t=2C_2\sqrt n \frac{c p}{\epsilon (p-q)}p^{t-1}$.
Then, from Lemma~\ref{lem:sensitivity}, we have
\[
\Delta(t) \leq 2C_2\sqrt n \gamma_t = 2C_2\sqrt n c q^{t-1}.
\]
Then, from $p\in(q,1)$, we have:
\[
\sum_{t=1}^\infty \frac{\Delta(t)}{M_t} 
\leq \frac{\epsilon(p-q)}{p}  \sum_{t=1}^\infty \left(\frac{q}{p}\right)^{t-1} 
= \frac{\epsilon(p-q)}{p}  \frac{p}{p-q} =\epsilon.
\] 
From Lemma~\ref{lem:lmc}, the algorithm guarantees $\epsilon$-differential privacy 
%with $\epsilon=\frac{2C_2 \sqrt n c_2q_1}{c_1(q_1-q_2)}$.
\end{proof}
%We can observe from the above theorem that $\frac{2C_2 \sqrt n c_2q_1}{c_1(q_1-q_2)} \rightarrow 0$ by letting $c_2\rightarrow 0$. 
%Thus, we conclude that by choosing the parameters ($c_1,c_2,q_1,q_2$) properly,  the iterative distributed algorithm guarantees any level of $\epsilon$-differential privacy. 

\subsection{Convergence}
\label{sec:con}
In this section, we prove that the algorithm converge.
We define the transfer matrix $\Phi(k,s) = \prod_{t=s+1}^k A(t)$, which captures the evolution of states under a sequence of communication graph $\{A_t\}_{s+1}^k$.
We denote $\Phi(k,s)_{i,j}$ as the entry of $\Phi(k,s)$ on the $i^{th}$ row and $j^{th}$ column.
The following lemma (Lemma 3.2 of~\cite{nedich09}) states that $\Phi(k,s)$ converges to a constant matrix as $k\rightarrow \infty$. Moreover, the convergence rate depends on: (i) the number of agents $N$, and (ii) the robust connectivity parameter $\eta$ given in Assumption~\ref{ass:graphrobust}.
\begin{lemma}
\label{lem:matconv}
Under Assumption~\ref{ass:graphrobust}, 
there exist constants $\theta>0$ and $\beta\in(0,1)$ such that for any $i,j\in[N]$ for any naturals $t>s$,
 \[
 |\Phi(t,s)_{i,j} -\frac{1}{N}|\leq \theta \beta^{t-s},
 \]
 where $\theta = \left( 1 - \frac{\eta}{4N^2}\right)^{-2}$ and $\beta =  1 - \frac{\eta}{4N^2}$.
\end{lemma}
\noindent
Lemma~\ref{lem:matconv} states a fundamental restriction on the rate of convergence given a communication topology. 
We can observe from the above lemma that: as the number of agents ($N$) grows or the robustness of communication ($\eta$) decreases,  $\beta$ becomes closer to 1, that is, the transition matrix ($\Phi$) converges slower.

Recall in Algorithm~\ref{alg:f}, agent $j$ influences agent $i$'s computation through the entry $a_{i,j}(t)$ of the communication graph $A_t$. 
Lemma~\ref{lem:matconv} states that any two agents $j$ and $k$ has the same longterm influence on agent $i$'s local state. 
%under our assumption of the robustness of connectivity of communication graphs.
As a direct result from this lemma, any two entries of $\Phi(t,s)$ converge to each other geometrically. That is, for any $i,j,k,l\in [N]$, 
$|\Phi(t,s)_{i,j}-\Phi(t,s)_{k,l}|\leq 2\theta \beta^{t-s}$.
For the algorithm defined by Algorithm~\ref{alg:r}-\ref{alg:u}, we compute the distance between any two local state
using the previous lemma. 

\begin{lemma}
\label{lem:con}
Under Assumptions~\ref{ass:function} and~\ref{ass:graphrobust}, for the proposed iterative distributed algorithm, 
for any agents $i,j\in[N]$ and any time $t\in\naturals$, the following holds:
\begin{equation}
\label{eq:lemcon}
\begin{array}{rl}
%||x_i(t)-x_j(t)||\leq& 2NC_1\theta\beta^t+2NC_2\theta\sum_{s=1}^t \gamma_s\beta^{t-s} \\
%& \ + 2N\theta\sum_{s=1}^t\beta^{t-s+1} ||w_k(s)||,
||x_i(t)-x_j(t)||\leq& M_1\beta^t+M_2\sum_{s=1}^t\beta^{t-s} \gamma_s \\
& \ + M_3\sum_{s=1}^t\beta^{t-s+1} ||w_k(s)||,
\end{array}
\end{equation}
where $\beta\in(0,1)$ is defined in Lemma~\ref{lem:matconv} and $M_1,M_2,M_3>0$ are bounded constants depends on the constants $C_1,C_2,C_3$ introduced in Assumption~\ref{ass:function}.
\end{lemma}
The proof can be found in appendix.
The above lemma bounds the distance between two agents' local states by three terms. The first term $M_1\beta^t$ decays to 0 as $t$ goes to infinity. The limits of the later two terms can be derived using Proposition~\ref{lem:convlim}. 
This lemma suggests that the limit of Equation~\eqref{eq:convbnd} depends on the limit of the noise magnitude as well as the limit of the step size. 
With Lemma~\ref{lem:con} and Proposition~\ref{lem:convlim}, the convergence of Algorithm described in Section~\ref{sec:alg} follows directly.

\begin{theorem}
\label{thm:con}
The algorithm described in Section~\ref{sec:alg} converges.
\end{theorem}
\begin{proof}
%The right-hand side of Equation~\eqref{eq:lemcon} has 3 terms: the order of the first term is $O(\beta^t)$,  that of the second term is $O(\sum_{s=1}^t \gamma_s\beta^{t-s})$ and that of the last term is $O(\sum_{s=1}^t\beta^{t-s+1} ||w_k(s)||)$.
%Submitting Equation~\eqref{eq:para} into the second term we get:
%\[
%\sum_{s=1}^t \gamma_s\beta^{t-s} \sim O( \sum_{s=1}^t q^{s}\beta^{t-s}).
%\]
%Notice that $\sum_{s=1}^t q^{s}\beta^{t-s} = \frac{q(q^t-\beta^t)}{q-\beta}$. 
%That is $\sum_{s=1}^t \gamma_s\beta^{t-s} \sim O(\max\{\beta^t,q^t\})$.
From Equation~\eqref{eq:para}, we have that
\[
\lim_{t\rightarrow \infty}\gamma_t = 0 \mbox{, and }  \lim_{t\rightarrow \infty}\ep ||w_k(t)|| =0. 
\]
Applying Proposition~\ref{lem:convlim}, we have
\[
\lim_{t\rightarrow \infty}\sum_{s=1}^t\beta^{t-s} \gamma_s=0 \mbox{, and } \lim_{t\rightarrow \infty} \sum_{s=1}^t\beta^{t-s+1} ||w_k(s)||=0.
\]
Then, by taking the limit of  Equation~\eqref{eq:lemcon}, we derive 
\[
\lim_{t\rightarrow \infty}\ep || x_i(t)-x_j(t)|| = 0.
\]
Thus the iterative distributed algorithm converges. %with rate $r$.
\end{proof}
Theorem~\ref{thm:con} shows that our proposed algorithm converges, which requires the expected distance between local values of different agents to converge to 0.
That is, the agents will eventually agree on a common value.
% as a solution of the optimization problem. 

\subsection{Accuracy}
\label{sec:acc}
In this section, we establish bounds on the accuracy of the proposed iterative distributed algorithm.
% discuss whether the local values eventually minimize the global cost function  by running the iterative distributed algorithm.
We first state a lemma which compares the sum of distance from $z_i(t)$ to any fixed point $x'$ to that of distance from $x_i(t)$ to $x'$.
\label{sec:cop}
\begin{lemma}
\label{lem:z-x-w}   
Fixed any point $x'\in \X$, for our proposed iterative distributed algorithm, for all $i\in[N]$, the following holds,
\begin{equation}
\label{eq:z-x-w0}
\sum_{i\in[N]}||z_i(t)-x'||^2\leq \sum_{i\in[N]}||x_i(t-1)-x'+w_i(t)||^2.
\end{equation}
\end{lemma}

We will derive a bound on the accuracy of the proposed algorithm (Theorem~\ref{thm:acc}).
% which shows that our proposed algorithm guarantees $d$-accuracy.
%We will derive a bound on $||x_i - x^*||$ using strong convexity and summing over all N we obtain … . Then using the <some facts baout expectations of w_t> and linearity of expectation we obtain
This bound is derived using Lemma~\ref{lem:z-x-w} and strong convexity. 
%and iteratively substituting Equation~\eqref{eq:mech1} and~\eqref{eq:mech2} into Equation~\eqref{eq:mech3}.
%%%%%%%%%%%%%%
%%The accuracy guarantee shown in Theorem~\ref{thm:acc} has a complicated expression. 
%Later we will show that the algorithm can be arbitrarily accurate by properly choosing the algorithm parameters. 

%\begin{theorem}
%\label{thm:convrate}
%
%\[
%2C_1e^{-\frac{C_3c_2(1-q_2^t)}{1-q_2}}++ \frac{\epsilon^2c^2}{12n} + \frac{2c^2}{1-q^2}.
%\]
%\end{theorem}

\begin{theorem}
\label{thm:acc}
The  algorithm guarantees $d$-accuracy with 
\begin{equation}
%d \sim O\left(e^{-\frac{c}{1-q}} \right) +O\left(\frac{c^2}{1-q^2} \right) + O\left(\frac{c^2p^2}{\epsilon^2(p-q)^2(1-p^2)}\right).
d =C_1e^{-\frac{C_3c}{1-q}} +\frac{C_2^2c^2}{1-q^2} + \frac{8C_2^2nc^2p^2}{\epsilon^2(p-q)^2(1-p^2)}.
\end{equation}
% where $\Psi(k,s) \deq \prod_{t=s+1}^k(1-C_3\gamma_t)$ is a scaler depends on $\gamma_t$.
\end{theorem}
Proof of the theorem can be found in appendix.
In the above theorem, we derived a bound of the accuracy the algorithm guarantees. 
%This bound has three terms. The third term $\frac{2c_1^2}{1-q_1^2} \rightarrow 0$ as $c_1\rightarrow 0$.
%The first term decreases if $\frac{c_2}{1-q_2}$ increases, while the second term decreases if $\frac{c^2_2}{1-q^2_2}$ decreases. 
This bound depends on the 4 parameters $\epsilon, c,p,q$. Fixing other three parameter,  the accuracy has the order of $d\sim O(\frac{1}{\epsilon^2})$.
As $\epsilon$ converges to 0, that is, for complete privacy for individuals, the accuracy becomes arbitrarily bad. 
%On the other hand, It can also be shown that as the accuracy parameter $d$ converges to 0, that is, for complete accuracy, the privacy parameter $\epsilon$ has to be chosen 

%\begin{lemma}
%\label{lem:arbacc}
%For any $d>0$ there exists a privacy level $\epsilon>0$ and a selection of algorithm parameters $c>0,q\in(0,1),p\in(q,1)$ such that
%the proposed algorithm guarantees $d$-accuracy.
%\end{lemma}
%The proof of this lemma can be found in appendix. Although the first term and the second term seems competing each other, this lemma shows that their sum can be reduced to arbitrary small. 
%That is, our algorithm  guarantees arbitrary level of accuracy if needed. 

\subsection{Experiment and Discussion}
\label{sec:disc}
The algorithm has four parameters: the privacy level $\epsilon$, the initial step size $c$, the step size decay rate $q$ and the noise decay rate $p$.
We have established that the algorithm guarantees $\epsilon$-differential privacy for any choice of  parameters.
If we fix the privacy level $\epsilon$, the dependency of the accuracy level of the algorithm on each of the other three parameters based on the partial derivative of $d$. 
%The dependency sometimes involves the relative magnitudes of constants $C_1$, $C_2$ and $C_3$.
%In that case, we assume that $C_1\gg C_2,C_3$, because in many applications we want to search for an optima over a relatively large domain.
Since the accuracy level the three parameters $d$ is not convex on $c,q,p$, the global optimal choice of the parameters does not have a clean close form expression.
However, we observe that if we fix any other two parameters,  the other parameter has a local optima:  
%the initial step size $c$ increases, the accuracy level first  increases and then decreases, where the best accuracy can be achieved at $c^*$ that solves the equation 
%\[C_1C_3e^{\frac{C_3 c^*}{1-q}}=2C_2^2c^*\left(\frac{1}{1-q^2}+ \frac{8np^2}{\epsilon^2(p-q)^2(1-p^2)}\right).\]
\begin{enumerate}[(I)]
% \item if the initial noise $c_1$ increases, then 
%	\begin{inparaenum}[(i)]
%	\item the privacy level increases, but 
%	\item the accuracy level decreases,
%	\end{inparaenum}
%\item  if the noise decaying decay $q_1$ increases, then
%	\begin{inparaenum}[(i)] 
%	\item the privacy level increases,  and
%	\item accuracy level decreases,
%	\end{inparaenum}
%\item  if the initial step size $c_2$ increases, 
%	\begin{inparaenum}[(i)]
%	\item the privacy level decreases, and 
%	\item the accuracy level increase first and then decreases. The best accuracy can be achieved if $c_2 = c^*$ which solves the equation $C_1C_3e^{\frac{C_3c^*}{q_2-1}}=\frac{C_2c^*}{1+q_2}$,
%	\end{inparaenum}
%\item   if the step decaying rate $q_2$ increases,
%	\begin{inparaenum}[(i)]
%	\item the privacy level decreases, and 
%	\item the accuracy level increase first and then decreases. The best accuracy can be achieved if $q_2 = q^*$ which solves the equation $C_1C_3e^{\frac{C_3c_2}{q^*-1}}=\frac{C_2^2c_2}{(1+q^*)^2}$.

    %\item if the privacy parameter $\epsilon$ decreases, which means more privacy, the accuracy level decreases,
    \item fixing parameters $q,p$, the best accuracy can be achieved at $c^*$ that solves the equation 
    \[C_1C_3e^{-\frac{C_3 c^*}{1-q}}=2C_2^2c^*\left(\frac{1}{1-q^2}+ \frac{8np^2}{\epsilon^2(p-q)^2(1-p^2)}\right),\]
    \item fixing parameters $c,p$, the best accuracy can be achieved at $q^*$ which solves the equation $\frac{C_1C_3 c}{1-q^{*2}}e^{-\frac{C_3c}{1-q^*}} =\frac{2q^{*2}C_2^2c^2}{(1-q^{*2})^2} + \frac{16C_2^2nc^2p^2}{\epsilon^2(p-q^*)^3(1-p^2)}$,
    \item fixing parameters $c,q$, the best accuracy can be achieved at
    $p^*$ which solves the equation $q(1-p^{*2}) = p^{*2}(p^*-q)$.
\end{enumerate}
In practice, we can tune the parameters with the following heuristic:
(i) pick $c,q,p$ randomly initially, (ii) fix two parameters and tune the remaining parameter to the local optima, and 
(iii) repeat step (ii) several times with different choice of parameters to be tuned.
We use the  proposed algorithm to solve Example~\ref{ex:problem} where the parameters ($c,q,p$)  are tuned with the above heuristic.
\begin{example}
\label{ex:solution}
\hili{
We solve a version of Example~\ref{ex:problem} with seven different privacy levels: $\epsilon = 0.1,0.2, 0.5,1, 2,5$ and $10$. 
We assign the domain of optimization $\X$ as the unit square $\X = \{(x,y)\in \reals^2 | -1 \leq x,y\leq 1\}$. 
For each privacy level $\epsilon$, we first decide the parameters $(c,q,p)$ using a heuristic, and then solve the DPOD repeatedly for $5000$ times. Each time, we record the  squared distance from the convergent point to the optima.  Then, the accuracy level $d$ of a privacy level is approximated  by the average of the squared distances over the 5000 runs.
The  experimental results are illustrated in Fig~\ref{fig:ex}.
}
\end{example}

\begin{figure}
\centering
\includegraphics[width = .3\textwidth]{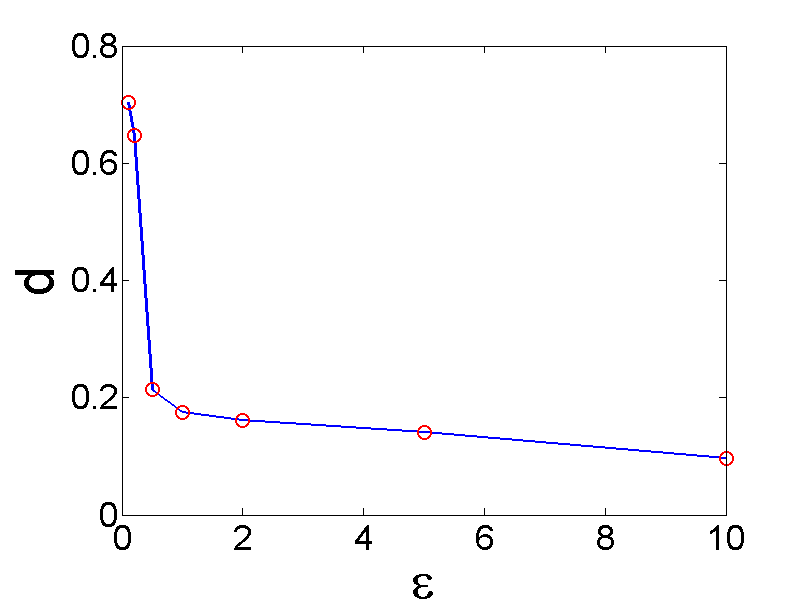}
\caption{Accuracy level $d$ as a function of privacy level $\epsilon$ for Example~\ref{ex:solution}.}
\label{fig:ex}
\end{figure}

\section{Conclusion}

We formulated the private distributed optimization (PDOP) problem in which $N$ agents are required to minimize a global cost function $f$ that is the sum $\Sigma_{i=1}^N f_i$ of $N$ cost functions for the individual agents. The agents may exchange information about their estimates for the optimal solution, but are required to keep their cost functions, namely the $f_i$'s, differentially private from an adversary with access to all the communication. We studied structurally simple iterative distributed algorithms for solving PDOP. Like other iterative algorithms for consensus and optimization, our algorithm proceeds in rounds. In each round, however, an agent first adds a vector of carefully chosen random noise to its current estimate for the optimal point and broadcasts this noisy estimate to its neighbors. The noise is chosen from a Laplace distribution that converges to the Dirac distribution with increasing number of rounds.In the second phase, the agent updates its estimate by (a) taking a weighted average of the noisy estimates it received from its neighbors and (b) moving the estimate, by a carefully chosen step-size,  in opposite direction of a the gradient of its own cost function ($f_i$ for agent $i$). The communication topology and hence the neighbors of an agent may change from one round to another, yet,  this structurally simple algorithm solves PDOP. We establish its differential privacy as well as its approximate convergence to the optimal point. The analysis also reveals the dependence of the accuracy and the privacy levels of the algorithm on the the noise and the step-size parameters. 
We observe that, by fixing other parameters, the accuracy level has the order of $O(\frac{1}{\epsilon^2})$.

Accurately solving distributed coordination problems require information sharing. Participants in such distributed coordination might be willing to sacrifice on the quality of the solution provided this loss is commensurate with the gain in the level of privacy of their individual preferences. Thus, a natural question is to quantify the cost or inaccuracy incurred in solving the problem as a function of the privacy level. In this paper, we have addressed this question in the context of PDOP and the class of iterative algorithms. Even for the class of iterative algorithms, establishing a lower-bound on the maximum level of differential privacy that can be achieved for a certain level of accuracy remains an open problem.

\bibliographystyle{abbrv}
\bibliography{Privacy,sayan1}

\newpage
\appendix
% !TEX root=main.tex
% 
% 
Proof of Proposition~\ref{lem:convlim}:
\begin{proof}
$\{a_k\}_{k=1}^\infty$ is a convergent sequence, thus bounded.
Let $|a_k|\leq M$ for all $k\in \naturals$.
Fixed any $\epsilon>0$, there exists an $N_1\in\naturals$ such that for all $k\geq N_1$, $|a_k|\leq \frac{\epsilon(1-\beta)}{2}$.
There exists an $N_2\in\naturals$ such that $\beta^{N_2}\leq \frac{\epsilon(1-\beta)\beta^{N_1}}{2M(\beta-\beta^{N_1})}$.
For all $n\geq \max\{N_1,N_2\}$, the absolute value of the summation in Equation~\eqref{eq:convlim} is bounded:
\begin{equation}
\label{eq:convbnd}
\begin{array}{rl} 
|\sum_{k=1}^n \beta^{n-k}a_k| &\leq \sum_{k=1}^{N_1-1} |\beta^{n-k}a_k| + \sum_{k=N_1}^{n} |\beta^{n-k}a_k| \\
&\leq  M\sum_{k=1}^{N_1-1} \beta^{n-k} + \sum_{k=N_1}^{n} \beta^{n-k}|a_k| 
\end{array}
\end{equation}
The first summation of  Expression~\eqref{eq:convbnd} is $\sum_{k=1}^{N_1-1} \beta^{n-k} = \frac{\beta^{n}(\beta-\beta^{N_1})}{\beta^{N_1}(1-\beta)}$. 
For $n\geq N_2$, we have $\beta^n\leq \frac{\epsilon(1-\beta)\beta^{N_1}}{2M(\beta-\beta^{N_1})}$. 
Thus, 
\[
M\sum_{k=1}^{N_1-1} \beta^{n-k} \leq M \frac{\epsilon(1-\beta)\beta^{N_1}}{2M(\beta-\beta^{N_1})}\frac{(\beta-\beta^{N_1})}{\beta^{N_1}(1-\beta)} \leq\frac{\epsilon}{2}.
\]
In the second summation of  Expression~\eqref{eq:convbnd}, we have $|a_k|\leq \frac{\epsilon(1-\beta)}{2}$ from the construction of $N_1$. Thus
\begin{eqnarray*}
\sum_{k=N_1}^{n} \beta^{n-k}|a_k| &\leq& \frac{\epsilon(1-\beta)}{2} \sum_{k=N_1}^{n} \beta^{n-k}\leq \frac{\epsilon(1-\beta)}{2} \sum_{i=0}^{\infty} \beta^{i} \\
& =&\frac{\epsilon(1-\beta)}{2} \frac{1}{1-\beta} = \frac{\epsilon}{2}.
\end{eqnarray*}
Substitute the above inequilities into Equation~\eqref{eq:convbnd}, we have $|\sum_{k=1}^n \beta^{n-k}a_k|\leq \epsilon$ for $n\geq \max\{N_1,N_2\}$.
Thus it follows that  $\lim \limits_{n\rightarrow \infty} \sum_{k=1}^n \beta^{n-k}a_k = 0$.
\end{proof}

Proof of Proposition~\ref{lem:uniqueinverse}:
\begin{proof}
Fixed $\P$, the communication graphs $A_t$ are fixed. Fixed  an observation sequence $\rho$, the messages $y(t)$ at each round $t$ are fixed.
From Equation~\eqref{eq:mech2},  for each $i\in[N]$ and $t\in \naturals$, $z_i(t)$ is uniquely determined. Then by Equation~\eqref{eq:mech3}, recalling that $f_i$ is specified by $\P$, we can conclude that $x_i(t)$ is uniquely specified for each $i\in[N]$ and $t\in\naturals$.
Thus, the execution $\alpha = x(0),\langle x(1),y(1),z(1)\rangle,\dots = \R^{-1}(\P,\rho,x(0))$ is uniquely determined.
\end{proof}
Proof of Lemma~\ref{lem:con}:
\begin{proof}
For brevity, we denote 
\[u_i(t) = \pj \X [z_i(t)- \gamma_t \grad f_i(z_i(t))] -z_i(t).\]
Then, we can rewrite Equation~\eqref{eq:mech3} as 
\begin{eqnarray*}
x_i(t)
%=z_i(t)+u_i(t)\\
= \sum_{j\in [N]} a_{ij}(t)x_j(t-1)+\sum_{j\in [N]} a_{ij}(t)w_j(t)+u_i(t).
\end{eqnarray*}
By the property of projection, we have
%\begin{eqnarray*}
%||x_i(t)-x_j(t)|| &\leq& ||\sum_{k\in [N]} (a_{ik}(t)-a_{jk}(t))x_j(t-1)+\sum_{k\in [N]} (a_{ik}(t)-a_{jk}(t))w_j(t)+\gamma_t(u_i(t)-u_j(t))|| \\
%&\leq& 
%\end{eqnarray*}
Recursively apply the above equation, we have:
\begin{equation}
\begin{array}{l}
x_i(t) = \sum_{j\in[N]}\Phi(t,0)_{i,j}x_j(0) +\sum_{s=1}^t\sum_{j\in[N]}\Phi(t,s)_{i,j}u_j(s) \\
 \qquad \qquad + \sum_{s=0}^{t-1}\sum_{j\in[N]} \Phi(t,s-1)_{i,j}w_j(s).
\end{array}
\end{equation}
Thus, the distance between two local states $x_i(t)$ and $x_j(t)$  is:
\begin{eqnarray*}
& &||x_i(t)-x_j(t)||\\
&=&\sum_{k\in[N]}|\Phi(t,0)_{i,k}-\Phi(t,0)_{j,k}| ||x_k(0)|| \\
& & \ +\sum_{s=1}^t\sum_{k\in[N]}|\Phi(t,s)_{i,k}-\Phi(t,s)_{j,k}| ||u_k(s)|| \\
&& \ + \sum_{s=0}^{t-1}\sum_{k=1}^N|\Phi(t,s-1)_{i,k}-\Phi(t,s-1)_{j,k}|||w_k(s)||. \\
\end{eqnarray*}
By applying Lemma~\ref{lem:matconv}, the above expression can be reduced to
\begin{eqnarray*}
& &||x_i(t)-x_j(t)|| \\
&\leq& 2N\theta \beta^t \sup_{k\in[N]}||x_k(0)||  + 2N\theta\sum_{s=1}^t \beta^{t-s}\sup_{k\in[N]}||u_k(s)|| \\
& & \   + 2N\theta\sum_{s=1}^t\beta^{t-s+1} ||w_k(s)||.
\end{eqnarray*}
From Assumption~\ref{ass:function}, we have $||x_k(0)||$ is bounded, and $||\grad f_k(z_k(s))|| \leq C_2$. From the property of projection,  $||u_k(s)||=||\pj \X[z_k(s)- \gamma_t\grad f_k(z_k(s))]-z_k(s)||\leq ||\gamma_t\grad f_k(z_k(s))|| \leq \gamma_tC_2$. Thus, we derive 
\begin{equation}
\label{eq:convbnd}
||x_i(t)-x_j(t)||
\leq M_1\beta^t+M_2\sum_{s=1}^t \gamma_s\beta^{t-s} + M_3\sum_{s=1}^t\beta^{t-s+1} ||w_k(s)||,
\end{equation}
where $M_1 =  2N\theta \sup \limits_{x\in\X}||x||, M_2 =2NC_2\theta$ and $M_3 = 2N\theta$.
\end{proof}
%%%
Proof of Lemma~\ref{lem:z-x-w}:
\begin{proof}
From Equation~\eqref{eq:mech1}-\eqref{eq:mech2}, we have 
$z_i(t)=\sum_{j\in[N]}a_{i,j}(t)(x_j(t-1)+w_j(t)).$
It follows that,
\begin{equation}
\label{eq:z-x-w}
\begin{array}{rl}
\sum_{i\in[N]}||z_i(t)-x'||^2 =& \sum_{i\in[N]} ||\sum_{j\in[N]}a_{i,j}(t)(x_j(t-1) \\
 & \ +w_j(t))  -x'||^2 
 \end{array}
\end{equation}
From the assumption that the matrix $A_{t}$ is doubly stochastic, we have$\sum_{j\in[N]}a_{i,j}(t)=1$. So we have $x' = \sum_{j\in[N]}a_{i,j}(t)x'$. Applying this trick to Equation~\eqref{eq:z-x-w}, we have
\begin{equation}
\label{eq:z-x-w1}
\sum_{i\in[N]}||z_i(t)-x'||^2 = \sum_{i\in[N]} ||\sum_{j\in[N]}a_{i,j}(t)\left(x_j(t-1)+w_j(t)-x'\right)||^2 .
\end{equation}
By triangle inequality and reordering of summation, we have
\begin{equation}
\label{eq:reorder1}
\begin{array}{rl}
&||\sum_{j\in[N]}a_{i,j}(t)\left(x_j(t-1)+w_j(t)-x'\right)||^2 \\
\leq& \sum_{i\in[N]}\sum_{j\in[N]}a_{i,j}(t) ||x_j(t-1)+w_j(t)-x'||^2. \\
=&\sum_{j\in[N]}\sum_{i\in[N]}a_{i,j}(t)||x_j(t-1)+w_j(t)-x'||^2.
\end{array}
\end{equation}
Again from the double stochasticity of $A_{t}$, $\sum_{i\in[N]}a_{i,j}(t)=1$. Then the above expression can be reduced to
\begin{eqnarray*}
&&\sum_{j\in[N]}\sum_{i\in[N]}a_{i,j}(t)||x_j(t-1)+w_j(t)-x'||^2  \\
&=& \sum_{j\in[N]}||x_j(t-1)+w_j(t)-x'||^2.
\end{eqnarray*}
Combining above equation with Equations~\eqref{eq:z-x-w1} and~\eqref{eq:reorder1}, we derive
\[
\sum_{i\in[N]}||z_i(t)-x'||^2\leq \sum_{j\in[N]}||x_j(t-1)-x'+w_j(t)||^2.
\]
By changing the variable of the right-hand side, the lemma follows.
\end{proof}

Proof of Theorem~\ref{thm:acc}:

\begin{proof}
From the property of stronly convex function, we have $\grad f_i(x))(y-x) \leq f_i(y)-f_i(x) - \frac{C_3}{2}||y-x||^2$ for any $x,y\in\X$. We denote $u_i(t) = -\grad f_i(z_i(t))$. Let $x^*$ be the minimum of the problem. Thus 
\begin{equation}
\label{eq:gradineq}
\begin{array}{rl}
u_i^T(t) (z_i(t)-x^*)\leq& f_i(x^*)-f_i(z_i(t))- \frac{C_3}{2}||z_i(t)-x^*||^2 
\\ \leq& - \frac{C_3}{2}||z_i(t)-x^*||^2.
\end{array}
\end{equation}
Take 2-norm on both side of Equation~\eqref{eq:mech3}, using the property of projection, we have
\begin{eqnarray*}
&&||x_i(t)-x^*||^2 \leq ||z_i(t)+\gamma_t u_i(t) -x^*||^2  \\
&=& ||z_i(t)-x^*||^2 + 2\gamma_t u_i^T(t)(z_i(t)-x^*)+ \gamma_t^2||u_i(t)||^2 .
\end{eqnarray*}
Combining this equation with Equation~\eqref{eq:gradineq} we have
\[
\begin{array}{c}
||x_i(t)-x^*||^2\leq   (1 - C_3\gamma_t)||z_i(t)-x^*||^2   +C_2^2\gamma_t^2.
\end{array}
\]
Sum up above equations over $i\in[N]$ and divided by $N$, we have
\begin{equation}
\label{eq:fromoptineq}
\frac{1}{N}\sum_{i\in[N]}||x_i(t)-x^*||^2\leq  \frac{1 - C_3\gamma_t}{N}\sum_{i\in[N]}||z_i(t)-x^*||^2  + C_2^2 \gamma_t^2.
\end{equation}
We will replace the terms $||z_i(t)-x^*||^2$ using Lemma~\ref{lem:z-x-w}. From Equation~\eqref{eq:z-x-w0}, we have:
\begin{eqnarray*}
&&\sum_{i\in[N]}||z_i(t)-x^*||^2 \leq \sum_{i\in[N]}||x_i(t-1)-x^*+w_i(t)||^2\\
&=&\sum_{i\in[N]}||x_i(t-1)-x^*||^2+\sum_{i\in[N]} [(x_i(t-1)-x^*)^Tw_i(t)] \\
 & & \ +\sum_{i\in[N]}||w_i(t)||^2
\end{eqnarray*}
 Under the condition $w_i(t)\sim Lap(M_{t})$, we have $\ep[w_i(t)]=0$ and $\ep||w_i(t)||^2=2M_{t}^2$. Noticing that $w_i(t)$ and  $x_i(t-1)$ are independent,   we have:
\begin{equation}
\label{eq:z-x}
\sum_{i\in[N]}\ep||x_i(t)-x^*||^2  \leq \sum_{i\in[N]} \ep||x_i(t-1)-x^*||^2 + 2NM_{t}^2.
\end{equation}
For simplicity we denote $S(t)\deq \frac{1}{N}\sum_{i\in[N]}\ep||x_i(t)-x^*||^2$.
Combining Equation~\eqref{eq:fromoptineq} and~\eqref{eq:z-x}, we have:
\begin{eqnarray}
\label{eq:indS}
S(t) \leq (1 - C_3\gamma_t)  S(t-1) + C_2^2\gamma_t^2 + 2(1 - C_3\gamma_t)M_{t}^2
\end{eqnarray}
Recursively apply Equation~\eqref{eq:indS}, we ultimately get:
\begin{equation}
\label{eq:19}
\begin{array}{rl}
S(t) \leq& \prod_{s=1}^t (1 - C_3\gamma_s) S(0) + C_2^2\sum_{s=1}^t \gamma_s^2\prod_{l=s+1}^t(1 - C_3\gamma_l) \\
& \ + 2\sum_{s=1}^t M_s^2\prod_{l=s}^t(1 - C_3\gamma_l).
\end{array}
\end{equation}
We define $\Psi(k,s) \deq \prod_{t=s+1}^k(1-C_3\gamma_t)$.
From Assumption~\ref{ass:function},we have that $S(0)\leq 2C_1$. 
Thus, we have 
\[
S(t)\leq 2C_1\Psi(t,0) + C^2_2 \sum_{s=1}^t \gamma_s^2\Psi(t,s) + 2\sum_{s=1}^t M_s^2 \Psi(t,s-1).
\]
The above equation has three terms, each of which involves $\Psi(k,s)$.We will give a bound to the term $\Psi(k,s)$.
Since $\Psi(k,s)$ is the product of factors no larger than $1$, $\Psi(k,s)\leq 1$ by definition. Thus, the above inequality reduces to
\[
\begin{array}{rl}
&S(t)\leq 2C_1\Psi(t,0) + C^2_2 \sum_{s=1}^t \gamma_s^2 + 2\sum_{s=1}^t M_s^2 \\
\leq& 2C_1\Psi(t,0) + C^2_2 \sum_{s=1}^\infty \gamma_s^2 + 2\sum_{s=1}^\infty M_s^2 \\
%\leq& 2C_1\Psi(t,0) + \frac{C^2_2c_2^2}{1-q_2^2} + \frac{2c^2}{1-q^2}.
\end{array}
\]
Substituting Equation~\eqref{eq:para} into the right-hand side, we have,
\[\gamma_t=cq^{t-1}\]
\[M_t =2C_2\sqrt n \frac{c p}{\epsilon(p-q)} p^{t-1}. \]
\begin{equation}
\label{eq:bndwops}
S(t)\leq 2C_1\Psi(t,0) + \frac{C_2^2c^2}{1-q^2} + \frac{8C_2^2nc^2p^2}{\epsilon^2(p-q)^2(1-p^2)}.
\end{equation}
To compute a tighter bound  of  term $\Psi(t,0)$,
we use a standard property of exponential function, that is, $1-a\leq e^{-a}$ for any $a\in \reals$. Thus
\[
\Psi(t,0) = \prod_{s=1}^t (1 - C_3\gamma_t) \leq e^{- \sum_{s=1}^t C_3\gamma_t} \leq e^{-\frac{C_3c(1-q^t)}{1-q}}.
\]
Substitute the above inequality into Equation~\eqref{eq:bndwops}, we have:
\[
S(t)\leq  2C_1e^{-\frac{C_3c(1-q^t)}{1-q}} +\frac{C_2^2c^2}{1-q^2} + \frac{8C_2^2nc^2p^2}{\epsilon^2(p-q)^2(1-p^2)}.
\]
By triangular inequality, we have 
\begin{eqnarray*}
& &\ep||\bar x(t)-x^*||^2 =  \ep||\frac{1}{N}\sum_{i\in[N]}x_i(t)-x^*||^2 \\
&\leq&\frac{1}{N}\sum_{i\in[N]}\ep||x_i(t)-x^*||^2=S(t)\\
%&\leq&  2C_1e^{-\frac{C_3c_2(1-q_2^t)}{1-q_2}}++ \frac{\epsilon^2c^2}{12n} + \frac{2c^2}{1-q^2}.
\end{eqnarray*}
Letting $t\rightarrow \infty$, we have
\[
\limsup_{t\rightarrow \infty} \ep||\bar x(t)-x^*||^2 \leq  C_1e^{-\frac{C_3c}{1-q}} +\frac{C_2^2c^2}{1-q^2} + \frac{8C_2^2nc^2p^2}{\epsilon^2(p-q)^2(1-p^2)}.
\]
Thus the theorem follows.
%It follows that 
%\begin{equation}
%\label{eq:xbartox}
%\ep||\bar x(t)-x^*||^2\leq  2C_1e^{-\frac{C_3c_2(1-q_2^t)}{1-q_2}}+\frac{C^2_2c_2^2}{1-q_2^2} + \frac{2c_1^2}{1-q_1^2}.
%\end{equation}
%By the property of a convex function $f$, 
%\begin{eqnarray*}
%f(\bar x(t)) &\leq& f^* + \grad f(\bar x(t))^T(x^*-\bar x(t)) \\
% &\leq& f^* + ||\grad f(\bar x(t))|| ||x^*-\bar x(t)|| \leq f^* + C_2||x^*-\bar x(t)||. 
%\end{eqnarray*}
%Taking expected value on both sides and combining with Equation~\eqref{eq:xbartox}, we have
%\[
%\ep [f(\bar x(t))] \leq f^*+ 2C_1C_2e^{-\frac{C_3c_2(1-q_2^t)}{1-q_2}}+\frac{C^3_2c_2^2}{1-q_2^2} + \frac{2C_2c_1^2}{1-q_1^2}.
%\]
%Letting $t\rightarrow \infty$, we get the following which proves the theorem:
%\[
%\lim_{t\rightarrow \infty} \ep [f(\bar x(t))] \leq 2C_1C_2e^{-\frac{C_3c_2}{1-q_2}}+\frac{C^3_2c_2^2}{1-q_2^2} + \frac{2C_2c_1^2}{1-q_1^2}
%\]
\end{proof}

%Proof of Lemma~\ref{lem:arbacc}:
%\begin{proof}
%The idea of the proof is to construct the following sequences: $c_1(n) =\frac{1}{n^2} $, $c_2(n)=\frac{1}{n^2}$, $q_1(n)=1-\frac{0.5}{n^3}$ and $q_2(n)=1-\frac{0.8}{n^3}$.
%It is clear that for each $n\in \naturals$, the following conditions hold: $c_1(n),c_2(n) > 0$, $q_1(n)\in(0,1)$ and $q_2(n)\in(0,q_1(n))$.
%We define
%\[
%d(n)\deq 2C_1C_2e^{-\frac{C_3c_2(n)}{1-q_2(n)}}+\frac{C^2_2c_2(n)^2}{1-q_2(n)^2} + \frac{2c_1(n)^2}{1-q_1(n)^2}.
%\]
%The right-hand side of the above equation can be reduced to
%\[
%d(n)=2C_1C_2e^{-\frac{C_3n^3}{n^2}}+\frac{C^2_2n^3}{0.8n^4(2-\frac{0.8}{n^3})} + \frac{2n^3}{0.5n^4(2-\frac{0.8}{n^3})}.
%\]
%It can be shown that $d(n)\rightarrow 0$ as $n\rightarrow \infty$. Thus, for any $d'$, there exists a $M\in \naturals$ such that $d(M)\leq d'$.
%Then, by choosing $c_1=c_1(M),c_2=c_2(M),q_1=q_1(M),q_2=q_2(M)$, we guarantees that the corresponding algorithm guarantees $d'$-accuracy.
%\end{proof}

\end{document}